\documentclass[12pt]{article}
\usepackage{amssymb}
\usepackage{multirow}
\usepackage[dvips]{graphics}
\usepackage{epsfig,rotating}
\usepackage{amssymb,amsmath}
\usepackage{latexsym}
\usepackage{tabularx}
\usepackage{array}
\usepackage{natbib}
\usepackage[retainorgcmds]{IEEEtrantools}
\usepackage{url}            
\usepackage{booktabs}       
\usepackage{amsfonts}       
\usepackage{nicefrac}       
\usepackage[protrusion=true,expansion=true]{microtype}      
\usepackage{lipsum}
\usepackage{amsmath}
\usepackage{amssymb}
\usepackage{indentfirst}
\usepackage{mathtools}
\usepackage{algorithm}
\usepackage{algpseudocode}
\usepackage{appendix}
\usepackage{hyperref} 
\hypersetup{
	colorlinks   = true, 
	urlcolor     = blue, 
	linkcolor    = blue, 
	citecolor   = blue 
}
\usepackage{nccmath}
\usepackage{centernot}

\renewcommand{\baselinestretch}{1.5}
\renewcommand{\appendixpagename}{\bf \huge Appendix}

\usepackage{bm,fullpage,enumerate,setspace}
	
\newcommand{\indep}{\perp \!\!\! \perp}
\newcommand{\nindep}{\not\!\perp\!\!\!\perp}

\begin{document}

\def\reduce#1{{\small  #1}}
\def\c#1{\ensuremath{\mathcal{#1}}}
\def\foot#1{\footnote{#1}}

\newcommand{\nmathbf}{\bm}

\def\bfA{\nmathbf A}
\def\bfB{\nmathbf B}
\def\bfC{\nmathbf C}
\def\bfD{\nmathbf D}
\def\bfE{\nmathbf E}
\def\bfF{\nmathbf F}
\def\bfG{\nmathbf G}
\def\bfH{\nmathbf H}
\def\bfI{\nmathbf I}
\def\bfJ{\nmathbf J}
\def\bfK{\nmathbf K}
\def\bfL{\nmathbf L}
\def\bfM{\nmathbf M}
\def\bfN{\nmathbf N}
\def\bfO{\nmathbf O}
\def\bfP{\nmathbf P}
\def\bfQ{\nmathbf Q}
\def\bfR{\nmathbf R}
\def\bfS{\nmathbf S}
\def\bfT{\nmathbf T}
\def\bfU{\nmathbf U}
\def\bfV{\nmathbf V}
\def\bfW{\nmathbf W}
\def\bfX{\nmathbf X}
\def\bfY{\nmathbf Y}
\def\bfZ{\nmathbf Z}

\def\bfa{\nmathbf a}
\def\bfb{\nmathbf b}
\def\bfc{\nmathbf c}
\def\bfd{\nmathbf d}
\def\bfe{\nmathbf e}
\def\bff{\nmathbf f}
\def\bfg{\nmathbf g}
\def\bfh{\nmathbf h}
\def\bfi{\nmathbf i}
\def\bfj{\nmathbf j}
\def\bfk{\nmathbf k}
\def\bfl{\nmathbf l}
\def\bfm{\nmathbf m}
\def\bfn{\nmathbf n}
\def\bfo{\nmathbf o}
\def\bfp{\nmathbf p}
\def\bfq{\nmathbf q}
\def\bfr{\nmathbf r}
\def\bfs{\nmathbf s}
\def\bft{\nmathbf t}
\def\bfu{\nmathbf u}
\def\bfv{\nmathbf v}
\def\bfw{\nmathbf w}
\def\bfx{\nmathbf x}
\def\bfy{\nmathbf y}
\def\bfz{\nmathbf z}

\def\bfalpha  {\nmathbf \alpha}
\def\bfbeta   {\nmathbf \beta}
\def\bfgamma  {\nmathbf \gamma}
\def\bfdelta  {\nmathbf \delta}
\def\bfepsilon{\nmathbf \epsilon}
\def\bfzeta   {\nmathbf \zeta}
\def\bfeta    {\nmathbf \eta}
\def\bftheta  {\nmathbf \theta}
\def\bfiota   {\nmathbf \iota}
\def\bfkappa  {\nmathbf \kappa}
\def\bflambda {\nmathbf \lambda}
\def\bfmu     {\nmathbf \mu}
\def\bfnu     {\nmathbf \nu}
\def\bfxi     {\nmathbf \xi}
\def\bfomicron{\nmathbf \omicron}
\def\bfpi     {\nmathbf \pi}
\def\bfrho    {\nmathbf \rho}
\def\bfsigma  {\nmathbf \sigma}
\def\bftau    {\nmathbf \tau}
\def\bfupsilon{\nmathbf \upsilon}
\def\bfphi    {\nmathbf \phi}
\def\bfpsi    {\nmathbf \psi}
\def\bfchi    {\nmathbf \chi}
\def\bfomega  {\nmathbf \omega}

\def\bfAlpha  {\nmathbf \Alpha}
\def\bfBeta   {\nmathbf \Beta}
\def\bfGamma  {\nmathbf \Gamma}
\def\bfDelta  {\nmathbf \Delta}
\def\bfEpsilon{\nmathbf \Epsilon}
\def\bfZeta   {\nmathbf \Zeta}
\def\bfEta    {\nmathbf \Eta}
\def\bfTheta  {\nmathbf \Theta}
\def\bfIota   {\nmathbf \Iota}
\def\bfKappa  {\nmathbf \Kappa}
\def\bfLambda {\nmathbf \Lambda}
\def\bfMu     {\nmathbf \Mu}
\def\bfNu     {\nmathbf \Nu}
\def\bfXi     {\nmathbf \Xi}
\def\bfOmicron{\nmathbf \Omicron}
\def\bfPi     {\nmathbf \Pi}
\def\bfRho    {\nmathbf \Rho}
\def\bfSigma  {\nmathbf \Sigma}
\def\bfTau    {\nmathbf \Tau}
\def\bfUpsilon{\nmathbf \Upsilon}
\def\bfPhi    {\nmathbf \Phi}
\def\bfPsi    {\nmathbf \Psi}
\def\bfChi    {\nmathbf \Chi}
\def\bfOmega  {\nmathbf \Omega}

\newcommand{\ttheta}{\tilde{\theta}}
\newcommand{\bfzero}{{\nmathbf 0}}
\newcommand{\bfone}{{\nmathbf 1}}
\newcommand{\vareps}{\varepsilon}
\def\bfvareps{\nmathbf \varepsilon}
\newcommand{\tgamma}{\tilde\gamma}

\newcommand{\cfA}{\mbox{\c{A}}}
\newcommand{\cfB}{\mbox{\c{B}}}
\newcommand{\cfC}{\mbox{\c{C}}}
\newcommand{\cfD}{\mbox{\c{D}}}
\newcommand{\cfE}{\mbox{\c{E}}}
\newcommand{\cfF}{\mbox{\c{F}}}
\newcommand{\cfG}{\mbox{\c{G}}}
\newcommand{\cfH}{\mbox{\c{H}}}
\newcommand{\cfI}{\mbox{\c{I}}}
\newcommand{\cfJ}{\mbox{\c{J}}}
\newcommand{\cfK}{\mbox{\c{K}}}
\newcommand{\cfL}{\mbox{\c{L}}}
\newcommand{\cfM}{\mbox{\c{M}}}
\newcommand{\cfN}{\mbox{\c{N}}}
\newcommand{\cfO}{\mbox{\c{O}}}
\newcommand{\cfP}{\mbox{\c{P}}}
\newcommand{\cfQ}{\mbox{\c{Q}}}
\newcommand{\cfR}{\mbox{\c{R}}}
\newcommand{\cfS}{\mbox{\c{S}}}
\newcommand{\cfT}{\mbox{\c{T}}}
\newcommand{\cfU}{\mbox{\c{U}}}
\newcommand{\cfV}{\mbox{\c{V}}}
\newcommand{\cfX}{\mbox{\c{X}}}
\newcommand{\cfY}{\mbox{\c{Y}}}
\newcommand{\cfZ}{\mbox{\c{Z}}}

\def\boldfacefake#1{\kern-4pt
   \hbox{ \mathsurround=0pt
   \hbox to 0.4pt{$#1$\hss}\hbox to 0.4pt{$#1$\hss}\hbox {$#1$}}}

\def\bfitI{\mbox{\boldfacefake{\it I}}}

\newcommand{\bcfA}{\boldsymbol{\mathcal{A}}}
\newcommand{\bcfB}{\boldsymbol{\mathcal{B}}}
\newcommand{\bcfC}{\boldsymbol{\mathcal{C}}}
\newcommand{\bcfD}{\boldsymbol{\mathcal{D}}}
\newcommand{\bcfE}{\boldsymbol{\mathcal{E}}}
\newcommand{\bcfF}{\boldsymbol{\mathcal{F}}}
\newcommand{\bcfG}{\boldsymbol{\mathcal{G}}}
\newcommand{\bcfH}{\boldsymbol{\mathcal{H}}}
\newcommand{\bcfI}{\boldsymbol{\mathcal{I}}}
\newcommand{\bcfJ}{\boldsymbol{\mathcal{J}}}
\newcommand{\bcfK}{\boldsymbol{\mathcal{K}}}
\newcommand{\bcfL}{\boldsymbol{\mathcal{L}}}
\newcommand{\bcfM}{\boldsymbol{\mathcal{M}}}
\newcommand{\bcfN}{\boldsymbol{\mathcal{N}}}
\newcommand{\bcfO}{\boldsymbol{\mathcal{O}}}
\newcommand{\bcfP}{\boldsymbol{\mathcal{P}}}
\newcommand{\bcfQ}{\boldsymbol{\mathcal{Q}}}
\newcommand{\bcfR}{\boldsymbol{\mathcal{R}}}
\newcommand{\bcfS}{\boldsymbol{\mathcal{S}}}
\newcommand{\bcfT}{\boldsymbol{\mathcal{T}}}
\newcommand{\bcfU}{\boldsymbol{\mathcal{U}}}
\newcommand{\bcfV}{\boldsymbol{\mathcal{V}}}
\newcommand{\bcfW}{\boldsymbol{\mathcal{W}}}
\newcommand{\bcfX}{\boldsymbol{\mathcal{X}}}
\newcommand{\bcfY}{\boldsymbol{\mathcal{Y}}}
\newcommand{\bcfZ}{\boldsymbol{\mathcal{Z}}}


\newcommand{\g}{\,\vert\,}
\newcommand{\p}{\mbox{P}}
\newcommand{\D}{\mbox{D}}
\newcommand{\E}{\mbox{E}}
\newcommand{\Mo}{\mbox{Mo}}
\newcommand{\Me}{\mbox{Me}}
\newcommand{\Cov}{\mbox{Cov}}
\newcommand{\Var}{\mbox{Var}}
\newcommand{\Corr}{\mbox{Corr}}
\newcommand{\Q}{\mbox{Q}}
\newcommand{\vech}{\mbox{vech}}
\newcommand{\tr}{\mbox{tr}}

\newcommand{\Bb}{\mbox{Bb}}
\newcommand{\Be}{\mbox{Be}}
\newcommand{\Bi}{\mbox{Bi}}
\newcommand{\Br}{\mbox{Br}}
\newcommand{\Ca}{\mbox{Ca}}
\newcommand{\Di}{\mbox{Di}}
\newcommand{\Ex}{\mbox{Ex}}
\newcommand{\Fs}{\mbox{Fs}}
\newcommand{\Ga}{\mbox{Ga}}
\newcommand{\Ge}{\mbox{Ge}}
\newcommand{\Gg}{\mbox{Gg}}
\newcommand{\Hy}{\mbox{Hy}}
\newcommand{\Ig}{\mbox{Ig}}
\newcommand{\Ip}{\mbox{Ip}}
\newcommand{\Lo}{\mbox{Lo}}
\newcommand{\Mu}{\mbox{Mu}}
\newcommand{\N}{\mbox{N}}
\newcommand{\Nb}{\mbox{Nb}}
\newcommand{\Ng}{\mbox{Ng}}
\newcommand{\Nw}{\mbox{Nw}}
\newcommand{\Pa}{\mbox{Pa}}
\newcommand{\Po}{\mbox{Po}}
\newcommand{\Pg}{\mbox{Pg}}
\newcommand{\Pn}{\mbox{Pn}}
\newcommand{\Ra}{\mbox{Ra}}
\newcommand{\St}{\mbox{St}}
\newcommand{\Un}{\mbox{Un}}
\newcommand{\Wi}{\mbox{Wi}}

\newcommand{\dd}[1]{\,d#1}
\newcommand{\barx}{\mbox{$\overline x$}}
\newcommand{\comb}[2]{{#1\choose#2}}
\newcommand{\ontop}[2]{{#1\atop#2}}
\newcommand{\h}{\hbox{$1\over2$}}
\newcommand{\ok}{\hfill\fbox{}}

\newcommand{\brow}[2]{\mbox{$\{{#1}_1,\ldots,{#1}_{#2}\}$}}
\newcommand{\prow}[2]{\mbox{$({#1}_1,\ldots,{#1}_{#2})$}}
\newcommand{\row}[2]{\mbox{${#1}_1,\ldots,{#1}_{#2}$}}
\newcommand{\data}{\row{x}{n}}
\newcommand{\bdata}{\prow{x}{n}}
\newcommand{\ie}{\emph{i.e.},\ }
\newcommand{\co}{\emph{cf.}\ }
\newcommand{\eg}{\emph{e.g.}, }
\newcommand{\etalc}{\emph{et al.},\ }
\newcommand{\etal}{\emph{et al.}\ }

\newenvironment{mat}{\left(\begin{array}}{\end{array}\right)}

\newcommand{\twomat}[5]{\begin{array}{ll}\displaystyle
              #1&\mbox{#2}\\[#5pt]#3&\mbox{#4}\end{array}}

\newcommand{\twocases}[6]
             {#1=\left\{\twomat{#2}{#3}{#4}{#5}{#6}\right.}

\newcommand{\mymatrix}[4]
           {\left(\begin{array}{ll}{#1} & {#2}\\
                                   {#3} & {#4}
                    \end{array} \right)}

\newcommand{\btable}{\begin{table}[h]\centering}
\newcommand{\etable}{\end{table}}
\newcommand{\bt}{\begin{parag}\small \let\b=\nsb \let\sb=\nssb \begin{tabular}}
\newcommand{\et}{\end{tabular}\let\b=\nb \let\sb=\nsb\end{parag}}
\newcommand{\capt}[1]
         {\begin{quotation}\caption{{\small #1 }}\vs{-2}\end{quotation}}

\newenvironment{parag}{\par}{\par}
\newenvironment{dif}
    {\begin{parag}\small \let\b=\nsb \let\sb=\nssb \begin{parag}}
    {\let\b=\nb \let\sb=\nsb \end{parag}\end{parag}}

\newenvironment{myproof}{\begin{dif} \noindent{\em Proof.~}}
            {\ok\vspace*{10pt}\end{dif}}
\newenvironment{exa}{\begin{list}{}
           {\setlength{\leftmargin}{10pt}
            \setlength{\rightmargin}{\leftmargin}}
           \item\begin{ex}\em}{\end{ex}\end{list}}

\newcommand{\be}{\begin{eqnarray}}
\newcommand{\ee}{\end{eqnarray}}
\newcommand{\ba}{\begin{eqnarray*}}
\newcommand{\ea}{\end{eqnarray*}}

\newcommand{\go}{\rightarrow}
\newcommand{\goi}{\rightarrow \infty}
\newcommand{\ul}{\underline}
\newcommand{\ol}{\overline}
\newcommand{\fr}{\frac}
\newcommand{\pn}{\par\noindent}
\newcommand{\nc}{\nonumber\\}
\newcommand{\ssum}{\mbox{$\sum$}}
\newcommand{\hhline}{\hline\hline}

\newtheorem{theorem0}{Theorem}
\newtheorem{lemma0}{Lemma}
\newtheorem{remark0}{Remark}
\newtheorem{fact0}{Fact}
\newtheorem{example0}{Example}
\newtheorem{definition0}{Definition}
\newtheorem{corollary0}{Corollary}
\newtheorem{proposition0}{Proposition}
\newtheorem{algorithmY}{Algorithm}

\newenvironment{theorem}{\begin{theorem0} \mbox{} }{\end{theorem0}}
\newenvironment{lemma}{\begin{lemma0} \mbox{}}{\end{lemma0}}
\newenvironment{remark}{\begin{remark0} \mbox{}}{\end{remark0}}
\newenvironment{fact}{\begin{fact0} \mbox{}}{\end{fact0}}
\newenvironment{example}{\begin{example0} }{\end{example0}}
\newenvironment{definition}{\begin{definition0} \mbox{}}{\end{definition0}}
\newenvironment{corollary}{\begin{corollary0} \mbox{} }{\end{corollary0}}
\newenvironment{proposition}{\begin{proposition0}\mbox{} }{\end{proposition0}}
\newenvironment{algorithm1}{\begin{algorithmY}\mbox{} }{\end{algorithmY}}

\newcommand{\reals}{\mbox{\rm I\kern-.20em R}}
\newcommand{\sreals}{\mbox{\small \rm I\kern-.20em R}}
\newcommand{\mylinel}{\renewcommand{\baselinestretch}{1.8}\tiny\small}
\newcommand{\goto}{\rightarrow}
\newcommand{\expect}{\E}

\newcommand{\bdfn}{\begin{dfn}}
\newcommand{\edfn}{\end{dfn}}
\newcommand{\bteo}{\begin{teo}}
\newcommand{\eteo}{\end{teo}}
\newcommand{\bexa}{\begin{exa}}
\newcommand{\eexa}{\end{exa}}
\newcommand{\bdif}{\begin{dif}}
\newcommand{\edif}{\end{dif}}
\newcommand{\bpro}{\begin{proof}}
\newcommand{\epro}{\end{proof}}


\thispagestyle{empty}

\begin{center}
{\bf\Large High-dimensional Statistical Inference and Variable Selection Using Sufficient Dimension Association
\\\vspace*{0.1in}}
\end{center}

\begin{center}
Shangyuan Ye$^1$\footnote{Corresponding author: sye@fiu.edu}, Shauna Rakshe$^3$, Ye Liang$^2$ 

$^{1}$Department of Mathematics and Statistics at Florida International University, Miami, FL 33119, U.S.A.\\
$^2$Department of Statistics, Oklahoma State University, OK 74078, U.S.A.\\
$^{3}$Biostatistics Shared Resource, Knight Cancer Institute, Oregon Health \& Science University,
Oregon, OR 97201, U.S.A.\\

\vspace*{0.1in}

\end{center}




\begin{quotation}
\small

\noindent {\bf Abstract} Simultaneous variable selection and statistical inference is challenging in high-dimensional data analysis. Most existing post-selection inference methods require explicitly specified regression models, which are often linear, as well as sparsity in the regression model. The performance of such procedures can be poor under either misspecified nonlinear models or a violation of the sparsity assumption. In this paper, we propose a sufficient dimension association (SDA) technique that measures the association between each predictor and the response variable conditioning on other predictors in the high-dimensional setting. Our proposed SDA method requires neither a specific form of regression model nor sparsity in the regression. Alternatively, our method assumes normalized or Gaussian-distributed predictors with a Markov blanket property. We propose an estimator for the SDA and prove asymptotic properties for the estimator. We construct three types of test statistics for the SDA and propose a multiple testing procedure to control the false discovery rate. Extensive simulation studies have been conducted to show the validity and superiority of our SDA method. Gene expression data from the Alzheimer Disease Neuroimaging Initiative are used to demonstrate a real application. 

\vspace*{0.15in}

\end{quotation}

\newpage
\setcounter{page}{1}

\section{Introduction}
As we venture further into the era of big data, the proliferation of expansive datasets presents a novel array of analytical complexities. High-dimensionality remains one of the most important complexities, where thousands of predictors are commonly available for only hundreds or even tens of samples. For example, the Alzheimer Disease Neuroimaging Initiative (ADNI) study collects longitudinal clinical, brain imaging, and gene expression data to support Alzheimer's Disease research. The ADNI Gene Expression Profile, a single dataset from ADNI, contains microarray data of 49,386 probes from a total of 745 different individuals. A critical task for such high-dimensional data analysis is to identify important features (e.g., probes) that are associated with the outcome of interest. 

A prevalent strategy for variable selection in high-dimensional data analysis is to use regularization-based regression methods, such as LASSO \citep{tibshirani1996regression}, SCAD \citep{fan2001variable}, MCP \citep{zhang2010nearly}, and many others. These methods generally prescribe explicitly defined regression models and assume sparsity in the regression models. Besides variable selection, post-selection inference has emerged as a significant research direction in the past decade. The goal of post-selection inference is to derive valid statistical inference by accounting for the uncertainty inherent in the selection \citep{kuchibhotla2022post}. For example, \cite{lee2016exact,tibshirani2018uniform,mccloskey2023hybrid} have explored the conditional selective inference approach, where the inference can be made by conditioning on the selection procedure. The efficacy of the conditional selective inference method is contingent upon the performance of the focused selection procedure.

Sufficient dimension reduction (SDR, \cite{cook1998regression}) is a dimension reduction method for variable selection in low-dimensional settings. Under this framework, dimension reduction is achieved by assuming that there exist low-rank subspaces of the original covariate space, or dimension reduction spaces, such that the outcome is independent of the covariates when conditioning on the projection of the covariates onto these subspaces. The sliced inverse regression (SIR) proposed by \cite{li1991sliced} is the most popular approach for SDR. In high-dimensional settings, the sparsity of eigenvectors in dimension reduction spaces is often assumed, with research primarily focusing on consistently estimating the central subspace, i.e., the smallest dimension reduction space  \citep{ni2005note,lin2018consistency,lin2019sparse,lin2021optimality}. However, studies on inference for each covariate in high-dimensional settings are still limited. \cite{zhu2006sliced} investigated the limiting distribution of SIR in fixed dimensions, while \cite{zhao2022testing} extended this to diverging dimensions and introduced a mirror statistic approach for false discovery rate (FDR) control based on data splitting.

This article introduces a novel statistical inference method for high-dimensional settings using SDR. We first examine the necessary conditions for a predictor to be part of the Markov blanket \citep{candes2018panning}, which is the minimal set of variables encapsulating the dependency between outcome and covariates. According to the derived conditions, we propose to make inference and select variables based on a measure named sufficient dimension association (SDA). Utilizing the assumption of multivariate normality and sparsity of the precision matrix, we propose a LASSO-based estimator for the SDA, which can be used to test the significance of each covariate separately. Contrary to most existing SDR methods, our proposed method does not require the central subspace to be consistently estimated. To test each covariate's membership in the Markov blanket, we construct a simple $\chi^2$ statistic and two other statistics based on the Kolmogorov-Smirnov (KS) and Cram{\"e}r-von-Mises (CvM) principles. A multiple testing procedure has been proposed to control the FDR. 

The proposed SDA method does not require any explicitly specified regression model as opposed to most existing post-selection inference methods. This method is practically simple to understand and implement. Despite the many sophisticated variable selection methods that have been developed for high-dimensional data, the concept of univariate association testing is still popular in scientific applications. The SDA enjoys such simplicity as it is merely a (conditional) association measure for each univariate predictor. 

The proposed SDA has a tie to the concept of partial correlation in the literature. A partial correlation refers to the correlation between two random variables after adjusting for the effect of a set of controlling variables \citep{baba2004partial}. Assuming a joint Gaussian distribution for the response and the covariates, \cite{buhlmann2010variable} proposed the partial trustfulness and a PC-simple algorithm. \cite{li2017variable} extended the method to elliptical linear regression models. \cite{alabiso2023high} studied the partial faithfulness for high-dimensional linear mixed-effects models and \cite{liu2018variable} considered variable selection for partial linear models. For variable screening, \cite{xia2021copula} proposed a copula-based partial correlation measure, \cite{lu2020model} considered the conditional distance correlation measure, and \cite{huang2022kernel} developed the kernel partial correlation measure. These existing methods focus on the variable selection consistency and the sure screening property, while statistical inference remains unclear.
On the inference side for high-dimensional linear models, \cite{gong2018efficient} developed a test based on the maximum of a sequence of partial correlations, and \cite{hemerik2021permutation} considered permutation tests based on partial or semipartial correlations. 

The rest of the article is organized as follows. Section \ref{subsec:model} outlines notations and assumptions used in this paper. Section \ref{subsec:VS} introduces the measure for sufficient dimension association. Section \ref{subsec:pc} discusses the relationship between the proposed SDA and the partial correlation measure. Section \ref{subsec:SIR} reviews the sliced inverse regression. Section \ref{subsec:est} presents the proposed estimator for SDA, while Section \ref{subsec:thm} delves into its theoretical properties. Sections \ref{subsec:se} and \ref{subsec:test} introduce standard error estimation and hypothesis testing methods, respectively. The finite sample performance of the proposed method is evaluated in Section \ref{sec:sim} through extensive simulations. In Section \ref{sec:real}, the method is applied to gene expression data from the ADNI study, focusing on identifying genes linked to Alzheimer's disease. We conclude the paper with a thorough discussion in Section \ref{sec:con}.

\section{Sufficient Dimension Association} \label{sec:SDR}
\subsection{Assumptions and notations} \label{subsec:model}
Throughout this article, the superscript $0$ is used to represent the true value of a given parameter. For any vector $\bfa \in \mathbb{R}^n$, $a_i$ denote the $i$-th coordinate of $\bfa$ for any $i \in \{1, \cdots, p\}$, $\bfa(\cfI)$ denote the subvector of $\bfa$ with coordinates of $\cfI$ for any $\cfI \subset \{1, \cdots, p\}$, and $\bfa_{-i}$ is used to denote the subvector of $\bfa$ excluding the $i$-th coordinate. Similarly, for any $n \times p$ matrix $\bfA$, $A_{ji}$ denotes the $(j,i)$-th element of $\bfA$ for any $j \in \{1, \cdots, n\}$ and $i \in \{1, \cdots, p\}$, $\bfA(\cfJ, \cfI)$ denotes the submatrix of $\bfA$ with rows of $\cfJ$ and columns of $\cfI$ for any $\cfJ \subset \{1, \cdots, n\}$ and $\cfI \subset \{1, \cdots, p\}$, and $\bfA_{-j, -i}$ is used to denote the submatrix of $\bfA$ excluding the $j$-th row and the $i$-th column. For any set $\cfS$, $|\cfS|$ denotes the cardinality of $\cfS$.  

Let $\bfX = (X_1, \cdots, X_p)^\top$ be a $p$-dimensional vector of predictors and $Y$ is the response. Then $\bfy = (Y_1, \cdots, Y_n)^\top$ denotes the vector of $n$ response observations and $\bfx = (\bfx_1, \cdots, \bfx_n)^\top$ denotes the corresponding covariate matrix, where $\bfx_j = (X_{j,1}, \cdots, X_{j,p})^\top$ is the covariates vector for subject $j$. We assume that $\bfX$ is normalized with mean zero and covariance matrix $\bfSigma$. We consider the semiparametric model
$Y = f(\bfb_1^\top \bfX, \cdots, \bfb_d^\top \bfX, \epsilon) = f(\bfB^\top \bfX, \epsilon)$, 
where $f(\cdot)$ is an arbitrary unknown function, $\bfb_1, \cdots, \bfb_d$ are unknown vectors, $\bfB = (\bfb_1, \cdots, \bfb_d)$, and $\epsilon$ is independent of $\bfX$ with mean zero. The linear space spanned by $\bfB$, denoted as col($\cfB$), is called a dimension reduction space. The intersection of all dimension reduction spaces, denoted as $\cfS_{Y|\bfX}$, is called the central subspace for the regression of $Y$ on $\bfX$ \citep{li1991sliced,cook1996graphics}. The $\cfS_{Y|\bfX}$ is, by definition, unique and can capture all information of $Y$ given $\bfX$. By assuming an elliptical distribution for $\bfX$ and the {\it linearity condition}:
\begin{description}
    \item[(C1)] Linearity condition: ~$\E(\bfa^\top \bfX | \bfb_1^\top \bfX, \cdots, \bfb_d^\top \bfX)$ is a linear combination of $\bfb_1^\top \bfX, \cdots, \bfb_d^\top \bfX$ for every $\bfa \in \mathbb{R}^p$,
\end{description}
\cite{li1991sliced} showed that $\bfSigma \cfS_{Y|\bfX} = \mbox{col}(\bfLambda)$, where $\bfLambda := \E \{\Var(\bfX | \bfY)\}$. In this article, we further assume that $\bfX$ follows multivariate Gaussian, i.e. $\bfX \sim N(0, \bfTheta^{-1})$, where $\bfTheta = \bfSigma^{-1}$ is the precision matrix of $\bfX$.

For high-dimensional settings, we also assume that the dependence between $Y$ and $\bfX$ can be characterized through a Markov blanket $\bfX(\cfA)$ \citep{candes2018panning}, i.e., 
\be \label{sparse}
Y \indep \bfX | \bfX(\cfA),
\ee
where $\cfA \subset \{ 1, \cdots, p \}$ is the minimal index set satisfying (\ref{sparse}). In this context, we aim to make inference about whether predictor $i$ belongs to the set $\cfA$, which translates to the following hypothesis testing problem: 
\be \label{indep test}
H_0: Y \indep X_i | \bfX_{-i} ~\text{versus}~ H_1: Y \nindep X_i | \bfX_{-i}.
\ee
The connection between the Markov blanket assumption (\ref{sparse}) and the conditional association test (\ref{indep test}) is discussed in the supplemental material.
Unlike other high-dimensional approaches such as knockoffs \citep{barber2015controlling} and selective inference \citep{taylor2018post}, we do not require a sparse regression model, i.e. $|\cfA| \ll p$. Rather, we assume that the precision matrix $\bfTheta$ is sparse, i.e. $I_i = |\cfI_i| \ll p$ where $\cfI_i = \{j: \Theta_{i, j} \neq 0\}$ for every $i \in \cfI$. Note that the predictor $X_j$ is conditionally independent with $X_i$ if $\Theta_{i, j}=0$. The sparsity of conditional dependency in covariates is commonly assumed in the SDR literature, examples include but not limit to \citep{tan2018convex,lin2018consistency,pircalabelu2021graph}. 



Under the linearity condition {\bf (C1)} and letting $\delta(Y) = \bfTheta \E(\bfX | Y)$ \citep[Theorem 3.1]{li1991sliced}, we have $\delta(Y) \in \cfS_{Y|\bfX}$, which also implies that
$\E\{ g(Y) \delta(Y) \} = \bfTheta \Cov\{ \bfX, g(Y) \} \in \cfS_{Y|\bfX}$, 
where $g(\cdot) \in \cfF$ and $\cfF$ is a sequence of transformation functions of $Y$. Thus, with a sequence of transformation functions $\{ g_h(\cdot) \}_{h=1}^H$, let $\bfbeta_{h} := \bfTheta \Cov\{ \bfX, g_h(Y) \}$ for every $h \in \cfH$, where $\cfH = \{ 1, \cdots H \}$, we have Span$(\bfbeta_1, \cdots, \bfbeta_H) \subseteq \cfS_{Y|\bfX}$. Similar to \cite{cook2006using,wu2011asymptotic,zhao2022testing}, we also assume the {\it coverage condition}:
\begin{description}
    \item[(C2)] Coverage condition: $\mbox{Span} (\bfbeta_1, \cdots, \bfbeta_H) = \cfS_{Y|\bfX}$ when $H > d$.
\end{description}

\subsection{Measure of sufficient dimension association} \label{subsec:VS}
To establish an association between each predictor $X_i$ and the outcome variable $Y$ while controlling all other predictors, we start from the dependence structure within $\bfX$. Let $\beta_{hi}$ be the $i$th element of $\bfbeta_h$ and we have
$\beta_{hi} = \bftheta_i^\top \Cov \{\bfX, g_h(Y)\}$, 
where $\bftheta_i$ is the $i$-th column of the precision matrix $\bfTheta$. Due to the property of multivariate Gaussian distribution, the conditional distribution of $X_i$ given all other predictors $\bfX_{-i}$ follows a Gaussian distribution,  
\be \label{conditional}
X_i | \bfX_{-i} \sim N(-\sigma_i^2 \bftheta_{i, -i}^\top \bfX_{-i}, \sigma_i^2),
\ee
where $\sigma_i^2 = \theta_{ii}^{-1} > 0$ and $\bftheta_{i,-i}$ denotes $\bftheta_i$ excluding the $i$th element. Thus, let $\bfzeta_i = -\sigma_i^2 \bftheta_{i,-i}$ and then we can write (\ref{conditional}) as a linear regression
\be \label{lm} 
X_i = \bfzeta_i^\top \bfX_{-i} + Z_i,
\ee
where $Z_i \sim N(0, \sigma_i^2)$. Since we have assumed that the precision matrix $\bfTheta$ is sparse, consequently, the induced $\bftheta_{i, -i}$ and $\bfzeta_i$ are also sparse, with $\theta_{ij} = \zeta_{ij} = 0$ if $j \notin \cfI_i$. 

Recall the assumption that the dependence between $Y$ and $\bfX$ is determined by a Markov blanket (\ref{sparse}). The problem is to identify the Markov blanket $\bfX(\cfA)$ and make inference for each individual predictor $X_i$. The following proposition reveals the link between the sufficient dimension method and the membership of $X_i$ to $\bfX(\cfA)$.
\begin{proposition} \label{prop1}
Assume that conditions (C1) and (C2) hold. Then $i \in \cfA^c$ if $\Cov(Z_i, g_h(Y)) = 0$ for all $h \in \cfH$, and $i \in \cfA$ if there exists $h \in \cfH$ such that $\Cov(Z_i, g_h(Y)) \ne 0$.
\end{proposition}
\begin{myproof}
The coverage condition suggests that all values in $\cfS_{Y|\bfX}$ can be written as a linear combination of $\bfbeta_1, \cdots, \bfbeta_H$, which implies that $i \in \cfA^c$ if and only if $\beta_{1i}^0 = \cdots = \beta_{Hi}^0 = 0$. Since $\sigma_i^2 > 0$, the scaled parameter $\bfnu_i := \sigma_i^2 \bfbeta_i$ inherits the property of $\bfbeta_i$. From (\ref{conditional}) and (\ref{lm}), denote $\nu_{hi}$ the $h$th component of $\bfnu_i$, we have
\ba
\nu_{hi} = \E(\sigma_i^2 \bftheta_i^\top \bfX g_h(Y)) = \E(Z_i g_h(Y)) = \Cov(Z_i, g_h(Y)). 
\ea
This is because $\bftheta_i^\top \bfX = \theta_{ii} X_i + \bftheta_{i,-i}^\top \bfX_{-i} = \sigma_i^{-2} (X_i + \sigma_i^2  \bftheta_{i,-i}^\top \bfX_{-i}) = \sigma_i^{-2} (X_i - \bfzeta_i^\top \bfX_{-i}) = \sigma_i^{-2} Z_i$. The proof is complete. 
\end{myproof}
Proposition \ref{prop1} suggests that to test whether $X_i$ belongs to the Markov blanket $\bfX(\cfA)$, i.e. the conditional dependence between $Y$ and $X_i$ given $\bfX_{-i}$, we can test the marginal association between $Y$ and $Z_i$ through the covariance $\Cov(Z_i, g_h(Y))$. 

We introduce the concept {\it sufficient dimension association}, as defined in the proof of Proposition \ref{prop1}: $\nu_{hi}= \Cov(Z_i, g_h(Y))$, for a sequence of transformation functions $g_h(\cdot)$, $h\in \cfH$. The SDA sequence of $\nu_{hi}$ is a measure of conditional association between an individual predictor $X_i$ and the outcome $Y$ given all other predictors. For this association measure, we make no assumption about the regression function $f$. However, we need to test a sequence of $H$ hypotheses, i.e. $H_0: \Cov(Z_i, g_h(Y)) = 0$ vs. $H_1: \Cov(Z_i, g_h(Y)) \ne 0$, to satisfy the coverage condition. As a special case, for linear regression models, the SDA can reduce to the correlation between $Z_i$ and $Y$. 

\subsection{Relationship to the partial correlation} \label{subsec:pc}
The SDA measure is related to the partial or semipartial correlation measures \citep{cohen2013applied}. The partial correlation is defined as 
\be \label{pC}
\rho_{Y X_i \cdot \bfX_{-i}} = \Corr\{Y - \E(Y | \bfX_{-i}), X_i - \E(X_i | \bfX_{-i})\},
\ee
and the semipartial correlation is defined as
\be \label{spC}
\rho_{Y (X_i \cdot \bfX_{-i})} = \Corr\{Y, X_i - \E(X_i | \bfX_{-i})\},
\ee
for which linear models are typically assumed for the conditional expectations $\E(Y | \bfX_{-i})$ and $\E(X_i | \bfX_{-i})$. The partial correlation is often used to measure the conditional dependence for Gaussian linear models, where the multivariate Gaussian assumption implies linearity for $\E(Y | \bfX_{-i})$ and $\E(X_i | \bfX_{-i})$. However, when the model is nonlinear, the partial correlation disagrees with the conditional correlation in general, known as the inconsistency \citep{baba2004partial,vargha2013interpretation}.

The inconsistency in nonlinear models suggests that a single linear measure is insufficient to capture the nonlinear conditional dependence between $Y$ and $X_i$. To address this issue, existing approaches relax or modify the linearity assumptions on $\E(Y | \bfX_{-i})$ and $\E(X_i | \bfX_{-i})$ \citep{huang2010testing,shah2020hardness}, or replace the Pearson’s correlation in (\ref{pC}) with nonlinear alternatives, such as rank-based \citep{xia2021copula}, distance-based \citep{wang2015conditional,lu2020model}, or kernel-based \citep{huang2022kernel} correlation measures. Indeed, \cite{shah2020hardness} argues that there does not exist a uniformly valid conditional independence test for all problems. Methods relying on the nonparametric regression \citep{huang2010testing,wang2015conditional} are not feasible in high-dimensional settings. Semiparametric approaches, such as \cite{huang2022kernel}, are more scalable for high-dimensional variable selection, but valid inference remains challenging. Methods using parametric approaches to adjust for the confounding effect $\bfX_{-i}$, such as \cite{xia2021copula}, may still be vulnerable to model misspecification.

Our proposed SDA measure has a unique advantage that the nonlinear relationship between $Y$ and $\bfX_{-i}$ can be unspecified while requiring that $\bfX$ be multivariate Gaussian, which is equivalent to a linear model assumption for $\E(X_i | \bfX_{-i})$. The key to address the inconsistency and capture the nonlinear association between $Y$ and $Z_i$ is to use a sequence of covariance measures $\Cov(Z_i, g_h(Y))$, which is distinct from existing approaches in the literature. Utilizing the theory of SDR, our proposed method is flexible for a wide class of models in high-dimensional settings. In a later section, we show that our testing procedure only requires fitting a single high-dimensional linear model, which is more computationally efficient than some tests based on high-dimensional partial correlation measures, for instance, the permutation test proposed by \cite{hemerik2021permutation}. Lastly, we note that for linear regression models, with $H = 1$ and $g_1(Y) = Y$, the SDA is equivalent to the semipartial correlation specified in (\ref{spC}). 

\subsection{Sliced inverse regression} \label{subsec:SIR}
Different choices of $\{ g_h(\cdot) \}_{h=1}^H$ have been proposed in the literature, examples of which can be found in \cite{yin2002dimension,cook2006using,wu2011asymptotic}. In this article, we focus on the sliced inverse regression (SIR) method proposed by \cite{li1991sliced}, where the response variable $Y$ is discretized into $H$ slices, i.e., $g_h(y) = I(y \in \cfJ_h)$, where $\cfJ_h$ is the set of all possible values for the $h$th slice. When $Y$ is a categorical random variable or only takes on a few values, each category or each unique value naturally defines a slice. When $Y$ is a continuous variable, the range of $Y$ is divided into $H$ slices based on a non-decreasing sequence $\{ a_h \}_{h=0}^H$ with $a_0 \le \min(Y)$, $a_H \ge \max(Y)$, and $\cfJ_h = \{ y: y \in (a_{h-1}, a_h) \}$. 

\section{Statistical Inference} \label{sec:inference}
\subsection{Target of estimation} \label{subsec:est}
We are interested in estimating the SDA sequence $\{\nu_{hi}, h\in\cfH\}$ and testing the sequence of hypotheses for SDA. Using the SIR technique, $\nu_{hi}$ can be expressed as $\nu_{hi} = \Cov\{ Z_i, I(Y \in \cfJ_h)\} = \E \{ I(Y \in \cfJ_h) Z_i \}$.
We propose an estimator for the SDA in the following form
\ba
\hat{\nu}_{hi} = \frac{1}{n} \sum_{j=1}^n I(Y_j \in \cfJ_h) (X_{i,j} - \hat{\bfzeta}_i^\top \bfX_{-i, j}),
\ea
where $\hat{\bfzeta}_i$ is an estimator of $\bfzeta_i$ from the linear model (\ref{lm}).

\subsection{Theoretical properties} \label{subsec:thm}
We first derive the necessary condition that $\hat{\bfnu}_i$, the vector of $\hat{\nu}_{hi}$, can be an asymptotic linear estimator (ALE).
\begin{lemma} \label{lemma1}
If $\lVert \hat{\bfzeta}_i - \bfzeta_i^{0} \rVert_1 = o_p( (H^2 \log p)^{-1/2} )$, then we have
\be \label{ALE}
\sqrt{n}(\hat{\bfnu}_{i} - \bfnu_{i}^0) - \frac{1}{\sqrt{n}} \sum_{j=1}^n \bfpsi_{i}(\bfw_j) = o_p(1),
\ee
where $\bfpsi_i(\bfW_j) = (\psi_{1i}(\bfW_j), \cdots, \psi_{Hi}(\bfW_j))^\top$ and $\psi_{hi} (\bfW) = I(Y \in \cfJ_h) (X_i - \bfzeta_i^{0,\top} \bfX_{-i}) - \nu_{hi}^0$, where $\bfW = \{ \bfX, Y \}$.
\end{lemma}
Lemma \ref{lemma1} requires that the $l_1$-norm of $\hat{\bfzeta}_i - \bfzeta_i^0$ converges faster than the order of $(H^2 \log(p))^{-1/2}$ as $p\goto\infty$. A proof of this lemma is provided in the supplemental material. 

Due to the sparsity assumption of $\bfzeta_i^0$, we consider the LASSO estimator \citep{tibshirani1996regression}, which minimizes the following penalized least squares:
\be \label{lasso}
\hat{\bfzeta}_i = \underset{\bfzeta_i \in \mathbb{R}^{p-1}}{\arg\min} \frac{1}{2n} \lVert \bfx_i - \bfzeta_i^\top \bfx_{-i} \rVert_2^2 + \lambda_i \lVert \bfzeta_i \rVert_1,
\ee
where $\lambda_i \ge 0$ is the tuning parameter. Theoretical properties of LASSO estimators have been extensively studied in the literature \citep{greenshtein2004persistence,meinshausen2006high,buhlmann2011statistics,bickel2009simultaneous}. By assuming the {\it restricted eigenvalue (RE) condition} on the design matrix $\bfx_{-i}$:
\begin{description}
    \item[(C3)] Restricted eigenvalue condition: Let $\cfC_a(\cfI) \subset \mathbb{R}^{p-1}$ be a set defined as $\cfC_a(\cfI) = \{ \bfb \in \mathbb{R}^{p-1}: \lVert \bfb(\cfI^c) \rVert_1 \le a \lVert \bfb(\cfI)\rVert_1 \}$, where $a > 0$. Then $\bfx_{-i}$ satisfies the restricted eigenvalue condition for $\kappa>0$ if $n^{-1} \lVert \bfb^\top \bfx_{-i} \rVert_2^2 \ge \kappa \lVert \bfb \rVert_2^2$, for every $\bfb \in \cfC_a(\cfI)$.
\end{description}
we have the following bound for the $l_1$-norm of $\hat{\bfzeta}_i - \bfzeta_i^0$:
\begin{lemma} \label{lemma2}
Under the RE condition (C3), when the tuning parameter $\lambda_i$ satisfies
\be \label{tuning0}
\lambda_i \ge \sqrt{\frac{C \sigma_i^2 \{ \log(p - 1) + \log(\delta) \}}{n}},
\ee
where $C>0$ and $\delta \rightarrow \infty$ as $n \rightarrow \infty$,
we have
\be \label{lasso error}
\lVert \hat{\bfzeta}_i - \bfzeta_i^{0} \rVert_1 = O_p(\lambda_i I_i).
\ee
\end{lemma}
Lemma \ref{lemma2} is a well-known result for the LASSO estimator \citep{bickel2009simultaneous,buhlmann2011statistics}. For the completion of our theoretical development, a proof of Lemma \ref{lemma2} is provided in the supplemental material. The error bound (\ref{lasso error}) relies on the RE condition, which is commonly assumed in the literature \citep{bickel2009simultaneous,wainwright2019high}. Intuitively, the condition regulates the Gram matrix $\bfx_{-i}^\top \bfx_{-i}/n$ so that the loss function would not be too flat at its minimizer. Although the RE condition is hard to verify in practice due to the unknown $I_i$, \cite{raskutti2010restricted} proved that the RE condition holds with a high probability for a broad class of Gaussian design matrices when the sample size satisfies $n = \Omega_p(I_i \log(p))$.

According to (\ref{lasso error}) and Lemma \ref{lemma1}, we assume the following regularity conditions:
\begin{description}
    \item[(C4)] The number of predictors $p$ satisfies $\lim_{n \rightarrow \infty} \log(p)/n \rightarrow 0$.
    \item[(C5)] The tuning parameter $\lambda_i$ satisfies (\ref{tuning0}) with $\delta \ll p$ and $I_i = o\left( \sqrt{n} (H \log(p))^{-1} \right)$.
\end{description}
Here, (C4) requires $p \ll e^n$ and (C5) imposes the sparsity condition. Notice that for a fixed $H$, the sparsity level becomes $o\left( \sqrt{n} / \log(p) \right)$, which is commonly assumed in the literature \citep{bickel2009simultaneous}. However, in more general settings, where $H$ is allowed to diverge, a more stringent assumption on the sparsity level is required. A direct corollary of Lemmas \ref{lemma1} and \ref{lemma2} shows that $\hat{\bfnu}_i$, under regularity conditions, is an ALE when the LASSO estimator is used.
\begin{corollary} \label{cor1}
Under regularity conditions (C3)-(C5), $\hat{\bfnu}_i$ is an asymptotic linear estimator when $\hat{\bfzeta}_i$ is the LASSO estimator of $\bfzeta_i$.
\end{corollary}

We then study the asymptotic distribution of $n^{-1/2} \sum_{j=1}^n \bfpsi_{i}(\bfw_j)$. For the high-dimensional setting, it is reasonable to assume that $d$, the rank of the central subspace $\cfS_{Y | \bfX}$, also increases with the sample size $n$. Thus, to satisfy the coverage condition (C2), the number of slices $H$ should also increase with $n$. We impose the following additional regularity conditions for the moments of $\bfnu_i$:
\begin{description}
    \item[(C6)] Let $p_h = P(Y \in \cfJ_h)$, there exist positive constants $\gamma_1 \le 1 \le \gamma_2$ such that the probability $p_h$ satisfies
    \ba
    \frac{\gamma_1}{H} \le p_h \le \frac{\gamma_2}{H} ~\text{for every $h \in \cfH$}.
    \ea
    \item[(C7)] There exist positive constants $\gamma_3$ and $\gamma_4$ such that, for every $h \in \cfH$, we have
    \ba
    |\E(Z_i | Y \in \cfJ_h)| < \gamma_3 H, ~ \E(Z_i^2 | Y \in \cfJ_h) < \gamma_4 H.
    \ea
\end{description}
Define $\bfOmega_i = \E\left[ \bfpsi_i(\bfW) \bfpsi_i^\top(\bfW) \right]$, where $\Omega_i(h,h) = p_h \E(Z_i^2 | Y \in \cfJ_h) - p_h^2 [\E(Z_i | Y \in \cfJ_h)]^2$ and $\Omega_i(h_1,h_2) = -p_{h_1} p_{h_2} \E(Z_i | Y \in \cfJ_{h_1}) \E(Z_i | Y \in \cfJ_{h_2})$,
for every $h, h_1, h_2 \in \cfH$. We first derive an upper bound for the approximating error of the multivariate Gaussian distribution to the distribution of $n^{-1/2} \sum_{j=1}^n \bfpsi_{i}(\bfw_j)$.

\begin{lemma} \label{lemma3}
Under regularity conditions (C6) and (C7), if $\bfOmega_i$ is invertible, there exists a constant $C$ such that
\ba
\sup_{A \in \cfC_H} \left| P\left( \frac{1}{\sqrt{n}} \sum_{j=1}^n \bfpsi_{i}(\bfw_j) \in A \right) - P\left( N(0, \bfOmega_i^0) \in A \right) \right| \le C \frac{H^{1/4}}{n^{1/2}},
\ea
where $\cfC_H$ is defined as the set of all convex subsets of $\mathbb{R}^H$.
\end{lemma}
Lemma \ref{lemma3} is a direct consequence of the multidimensional Berry-Esseen Bound \citep{bentkus2005lyapunov}, and a detailed proof is provided in the supplemental material. Finally, combining Corollary \ref{cor1} and Lemma \ref{lemma3}, we have proved the following theorem.

\begin{theorem} \label{thm1}
Under regularity conditions (C3)-(C7), if the number of slices $H$ satisfies $\lim_{n \rightarrow \infty} H n^{-2} \rightarrow 0$, and $\bfOmega_i$ is invertible, then
$\sqrt{n} (\hat{\bfnu}_{i} - \bfnu_{i}^0) \xrightarrow{d} N(0, \bfOmega_i^0).$
\end{theorem}

In ultra-high dimensional settings, where $\log(p) = O(n^a)$ for some $a \in (0, 1-2\kappa)$ and $\kappa \ge 0$, the average in-sample prediction error of the LASSO estimator (\ref{lasso}) will diverge as $n \rightarrow \infty$. Therefore, a sure independence screening (SIS) procedure, as proposed by \cite{fan2008sis}, should be employed to reduce the dimensionality to a moderate scale, to satisfy condition (C4). Specifically, for a given $\gamma \in (0, 1)$, we rank the absolute values of the pairwise Pearson correlation coefficients between $X_i$ and each $X_j$, i.e. $\hat{\rho}_{ij} = \widehat{\Corr}(X_i, X_j)$, in a decreasing order, and then select a subset of variables defined as
\ba
\cfM_{i \gamma} = \{ j \in \{1, \cdots, i-1, i+1, \cdots, p\}: \text{$\lfloor \gamma (n-1) \rfloor$ top-ranked $|\hat{\rho}_{ij}|$} \},
\ea
where $\lfloor \cdot \rfloor$ is the floor function. Under the following additional regularity conditions:
\begin{description}
    \item[(C8)] For some $\kappa \ge 0$ and $c_1, c_2 > 0$,
    \ba
    \min_{j \in \cfI_i} |\zeta_{ij}| \ge \frac{c_1}{n^\kappa} ~\text{and}~ \min_{j \in \cfI_i} |\Cov(\zeta_{ij}^{-1} X_i, X_j)| \ge c_2.
    \ea
    \item[(C9)] There exist $0 \le \tau < 1-2\kappa$ and $c_3 > 0$ such that $\lambda_{\max} (\bfSigma_{-i,-i}) \le c_3 n^\tau$.
\end{description}
\cite{fan2008sis} demonstrate that when $\gamma$ is chosen such that $\lfloor \gamma (n-1) \rfloor = O(n^{1-\upsilon})$ with $\upsilon < 1 - 2\kappa - \tau$, we have $P(\cfI_i \subset \cfM_{i \gamma}) \rightarrow 1 ~\text{as}~ n \rightarrow \infty$. Thus, the theoretical properties proved in this section remain valid in ultra-high dimensional settings when substituting $\bfX_{-i}$ with $\bfX(\cfM_{i \gamma})$ in (\ref{lm}). The first part of condition (C8) imposes a lower bound on the nonzero coefficients in model (\ref{lm}), which ensures that the probability of $\cfI_i \not\subset \cfM_{i\gamma}$ converges to zero. The second part of (C8) excludes the scenario where a covariate $X_j$ is marginally uncorrelated but conditionally correlated with $X_i$. Condition (C9) requires that the covariates are not excessively correlated. Together with the RE condition in (C3) and the sparsity condition in (C5), the first part of (C8) and condition (C9) are typically satisfied in practice. However, if the second part of (C8) is violated, the iterative SIS (ISIS) procedure proposed by \cite{fan2008sis} can be used as an alternative. 

\subsection{Standard error estimation} \label{subsec:se}
As shown in Theorem \ref{thm1}, the asymptotic variance of $\sqrt{n} (\hat{\bfnu}_{i} - \bfnu_{i}^0)$ is given by $\bfOmega_i$, which is defined as $\E\left[ \bfpsi_i(\bfW) \bfpsi_i^\top(\bfW) \right]$. Denote $\hat{\bfpsi}_i(\bfW) = (\hat{\psi}_{1i} (\bfW), \cdots, \hat{\psi}_{Hi} (\bfW))^\top$ to be the estimated influence vector, where $\hat{\psi}_{hi}(\bfW) = I(Y_i \in \cfJ_h) (X_i - \hat{\bfzeta}_i^\top \bfX_{-i}) - \hat{\nu}_{hi}$. We consider the following sample mean estimator:
\ba
\hat{\bfOmega}_i = \frac{1}{n} \sum_{j=1}^n \hat{\bfpsi}_i(\bfW_j) \hat{\bfpsi}_i^\top(\bfW_j). 
\ea
The consistency of $\hat{\bfOmega}_i$ is given in Theorem \ref{thm2}, and a detailed proof is provided in the supplemental material.
\begin{theorem} \label{thm2}
Under regularity conditions (C3)-(C7), we have $\hat{\bfOmega}_i \xrightarrow{P} \bfOmega_i^0$.
\end{theorem}

\subsection{Hypothesis testing} \label{subsec:test}
\subsubsection{Chi-squared test}
According to Proposition \ref{prop1}, the conditional independence test (\ref{indep test}) is equivalent to
\be \label{Hypothesis}
H_0: \bfnu_i = \mathbf{0} ~\text{versus}~ H_1: \bfnu_i \ne \mathbf{0}.
\ee
With asymptotic results in Sections \ref{subsec:thm} and \ref{subsec:se}, consider the Wald chi-squared test statistic $\hat{T}_i^{\chi} = \hat{\bfnu}_i^\top (\hat{\bfOmega}_i/n)^{-1} \hat{\bfnu}_i$. The following corollary provides the asymptotic distribution of $\hat{T}_i^{\chi}$.
\begin{corollary} \label{cor2}
Under regularity conditions (C1)-(C7), 
\ba
\hat{T}_i^{\chi} \xrightarrow{d} \chi^2_H (\nu_i^{0,\top} \bfOmega^{-1} \nu_i^0),
\ea
where $\chi^2_H(\lambda)$ is the noncentral chi-squared distribution with $H$ degrees of freedom and noncentral parameter $\lambda$.
\end{corollary}
Corollary \ref{cor2} is an immediate consequence of Theorem \ref{thm1} and \ref{thm2}. Under $H_0$, the asymptotic distribution of $\hat{T}_i^{\chi}$ is a central chi-squared distribution with $H$ degrees of freedom. We denote the hypothesis testing approach based on the chi-squared statistic by SDA-$\chi^2$.

\subsubsection{Kolmogorov-Smirnov and Cram{\"e}r-von-Mises tests}
Alternatively, because the test of hypothesis (\ref{Hypothesis}) is equivalent to $H_0: \nu_{hi} = 0$, for all $h \in \cfH$, versus $H_1: \nu_{hi} \ne 0$, for some $h \in \cfH$, we can consider each univariate test and then combine. The asymptotic normality in Theorem \ref{thm1} suggests a test statistic $z_{hi} = \sqrt{n} \hat{\nu}_{hi}/\sqrt{\hat{\Omega}_i(h, h)}$ and $z_{hi} \xrightarrow{d} N(0, 1)$ when $\nu_{hi} = 0$. Here, we consider a Kolmogorov-Smirnov (KS) type statistic:
\be \label{KS}
\hat{T}_i^{\text{KS}} = \max_{h \in \cfH} |z_{hi}| = \max_{h \in \cfH} \left| \frac{\sqrt{n} \hat{\nu}_{hi}}{\sqrt{\hat{\Omega}_i(h, h)}} \right|.
\ee
Thus, we reject $H_0$ if  $\hat{T}_i^{\text{KS}} > c_i$ where $c_i$ is some critical value. We also develop a Cram{\"e}r-von-Mises (CvM) type statistic $\hat{T}_i^{\text{CvM}} = \int |z_{hi}| d Q(h)$, where $Q$ is a weight function on $\cfH$. When using equal weights, the CvM-type statistic is defined as
\be \label{CvM}
\hat{T}_i^{\text{CvM}} = \frac{1}{H} \sum_{h=1}^H |z_{hi}| = \frac{1}{H} \sum_{h=1}^H \left| \frac{\sqrt{n} \hat{\nu}_{hi}}{\sqrt{\hat{\Omega}_i(h, h)}} \right|.
\ee
We name the proposed tests based on (\ref{KS}) and (\ref{CvM}) as SDA-KS and SDA-CvM, respectively. 

Because the asymptotic distributions of $\hat{T}_i^{\text{KS}}$ and $\hat{T}_i^{\text{CvM}}$ are difficult to derive analytically, we adopt a simulation-based approach referred to as the multiplier bootstrap (MB) \citep{van2000weak,chernozhukov2013gaussian}. The theoretical development of this method and a pseudocode are provided in the supplemental material.  

\subsection{Multiple hypothesis testing} \label{subsec:MH}
We propose an approach similar to the knockoff filter \citep{barber2015controlling,candes2018panning} for multiple hypothesis testing with FDR control. To generate a knockoff copy of $\bfX$, denoted as $\tilde{\bfX}$, there are two requirements: (a) each $\tilde{X}_i$ in $\tilde{\bfX}$ is exchangeable with the corresponding $X_i$ in $\bfX$; and (b) $\tilde{\bfX}$ provides no further regression information about the response $Y$, i.e. $\bfX$ and $\tilde{\bfX}$ need to be as dissimilar as possible. Generating knockoff copies can be challenging as it requires the distribution of $\bfX$ to be completely known \citep{candes2018panning} or can be consistently estimated \citep{barber2020robust}. However, our proposed SDA has an advantage due to the standardized variable $Z_i$. Thus, we can generate $\tilde{\bfz}_i$, a random sample of size $n$ from $N(0, \hat{\sigma}_i^2)$, which will suffice.    

Let $\tilde{T}_i$ be the test statistic calculated using the knockoff copy $\tilde{\bfz}_i$, we obtain the feature statistic $M_i = M(\hat{T}_i, \tilde{T}_i)$, where $M(\cdot, \cdot)$ is an antisymmetric function. Following Corollary \ref{cor1}, the distribution of $M_i$ is asymptotically symmetric to 0 for $i \in \cfA^c$ (more details are available in the proof of Theorem \ref{thm3}). Therefore, with a sufficient number of hypotheses, given a threshold value $t$, $\#\{i: M_i \le -t \}$ can be used as a conservative estimate of the number of false selections, regardless of the exact distribution of $M_i$ under the null. Thus, with a desired FDR level $q$, the set of selected variables is defined as $\{ i: M_i \ge \tau \}$, where the data-dependent threshold $\tau$ is defined as
\be \label{adthresh}
\tau = \min \left\{t: \frac{\#\{ i: M_i \le -t \}}{\#\{ i: M_i \ge t \}} \le q \right\}.
\ee
The pseudocode for the implementation of the proposed multiple hypothesis testing procedure is summarized in Algorithm \ref{algo2}. The next lemma, with a detailed proof available in the supplementary material, shows that both false positive proportion (FDP), which is defined as the proportion of false selections among all selected variables, and FDR can be controlled asymptotically using the threshold $\tau$ defined in (\ref{adthresh}).  
\begin{theorem} \label{thm4}
Under (C1)-(C7), let $p_0 = |\cfA^c|$, as $n \rightarrow \infty$ and $p_0 \rightarrow \infty$, the SDA procedure in Algorithm \ref{algo2} satisfies
\ba
P(\mbox{FDP}(\tau) \le q) \rightarrow 1 ~\text{ and }~ \limsup_{p_0 \rightarrow \infty} \text{FDR}(\tau) \le q.
\ea    
\end{theorem} 


\begin{algorithm}
\caption{False discovery rate control via SDA.} 
\label{algo2}
\begin{algorithmic}[1]
\State Divide the range of $Y$ into $H$ slices
\For {$i = 1, \cdots, p$}
    \State Obtain $\hat{\bfz}_i$ and $\hat{\sigma}_i^2$ by fitting a high-dimension regression model $X_i = \bfzeta_i^\top \bfX_{-i} + Z_i$
    \State Generate a knockoff copy $\tilde{\bfz}_i$ by randomly drawing $n$ sample from $N(0, \hat{\sigma}^2)$
    \State Calculate the test statistic $\hat{T}_i$ and $\Tilde{T}_i$ from $\hat{\bfz}_i$ and $\tilde{\bfz}_i$, respectively
    \State Calculate the feature statistic $M_i$
\EndFor
\State Calculate the data-dependent threshold $\tau$ defined in (\ref{adthresh})
\For {$i = 1, \cdots, p$}
    \If{$M_i > \tau$}
    \State Reject $H_i$
    \Else 
    \State Do not reject $H_i$
    \EndIf
\EndFor
\end{algorithmic}
\end{algorithm}

\section{Simulation Studies} \label{sec:sim}
In this section, we investigate the empirical performance of the proposed SDA-$\chi^2$, SDA-KS, and SDA-CvM procedures through extensive simulation scenarios. 

\subsection{Simulation settings} \label{subsec:set}
We set the significance level at $0.05$ and FDR at $0.1$ for multiple hypothesis testing. We consider $n = 200$ or $400$, and $p = 1000$ or $2000$. Tuning parameters of $\hat{\bfzeta}_i$ for all methods are selected based on ten-fold cross-validation.

\subsubsection{Correlation structures}
{\it Fixed precision matrix.} 
We first consider the setting with fixed correlation structures. We generate covariates $\bfX$ from a multivariate Gaussian distribution with mean zero and precision matrix $\bfTheta$. Here, we let $\bfTheta$ be a block diagonal matrix with cluster sizes of $q = 5$. Within each block, we let $\Theta_{ii} = 1$ and $\Theta_{ii'} = 0.5$ for $i \ne i'$. We denote $\cfB = \{1, \cdots, B \}$, where $B = p/q$, the index set of blocks.  

{\it Random network covariates.}
We then consider the scenario where the correlation structure of $\bfX$ is determined by a randomly generated small-world network \citep{watts1998collective}. Specifically, each $X_i$ is connected to covariates within $e=5$ neighbors, with a rewiring probability of 0.25. For each connected pair $(X_i, X_{i'})$, the corresponding entry $\theta_{ii'}$ in the precision matrix is uniformly sampled from $(-1, -0.5) \cup (0.5, 1)$.  

\subsubsection{Regression functions} 
We consider the following two single-index models (1 and 2) and two multiple-index models (3 and 4):
\ba
&\text{Model 1}& ~ Y = \bfb^\top \bfX + \epsilon; \\ &\text{Model 2}& ~ Y = \sin(\bfb^\top \bfX) \exp(\bfb^\top \bfX) + \epsilon; \\ &\text{Model 3}& ~ Y = \frac{3 \bfb_1^\top \bfX}{0.5 + (1.5 \bfb_2^\top \bfX)^2} + \epsilon; \\ &\text{Model 4}& ~ Y = \sum_{k=1}^d (0 \vee \bfb_k^\top \bfX) + \epsilon,
\ea
where $\epsilon \sim N(0, 1)$ and $\bfX \indep \epsilon$. 

\subsection{Simulation results}
We discuss key findings of simulation studies from multiple perspectives in the following five sub-sections. Detailed simulation settings are included in the supplemental material. 
\subsubsection{The choice of $H$} \label{subsec:result}
We first investigate the impact of the choice of $H$. Results for SDA-CvM, SDA-$\chi^2$, and SDA-KS are presented in Figures \ref{FigH}, S1, and S2, respectively. Across all regression models, the empirical type I error rates for the null variables are nearly unaffected by the choice of $H$. In the two single-index models (models 1 and 2), the empirical power for $b_1$ (the larger effect size) is consistently 1 across all values of $H$. For $b_2$ (the smaller effect size), the power decreases with increasing $H$ when $n = 200$, but this decreasing trend is less pronounced when $n = 400$. In model 3, the two active variables in $\bfb_1$ show patterns similar to those in models 1 and 2, while the two active variables in $\bfb_2$ behave differently. Specifically, the empirical power for $X_3$ (larger effect size) increases sharply from the null level at $H = 2$ and then stabilizes, whereas for $X_4$ (smaller effect size), the power increases and then gradually decreases. In model 4, although it is a multiple-index nonlinear model, the regression function is close to a linear one, leading to a decreasing trend in power across all $X_i$ as $H$ increases.

In summary, the optimal choice of $H$ depends on the sample size, the effect size, and the form of regression functions. Values of $H$ between 4 and 7 provide robust performance across all settings considered in this study. Therefore, we set $H = 5$ for the remaining simulation studies.

\subsubsection{Empirical selection rates} \label{subsec:result2}
We then study and compare the empirical type I error rates and power for SDA-$\chi^2$, SDA-KS, and SDA-CvM. To compare with existing methods, we also include the selective inference (SI) method \citep{lee2016exact, taylor2018post} and the high-dimensional permutation (HP) test based on the partial correlation \citep{hemerik2021permutation}. Empirical selection rates for the first 100 covariates are calculated based on 1000 simulated data sets. 

The SI method performs poorly in nonlinear settings due to very low selection rates of active variables under the LASSO estimator, as shown in Table S1. Since SI only applies to variables selected by LASSO, the low selection rates also result in low empirical power. Therefore, we focus our comparison on the proposed SDA methods and the HP method. Tables \ref{table_compare} and S2 summarize the simulation results for the fixed and network precision matrix structures, respectively. All methods conservatively control the type I error across all settings, with the exception of SDA-CvM and SDA-$\chi^2$ under the fixed precision matrix in model 1. The empirical power of all methods increases with larger sample sizes, stronger effect sizes, or larger cluster sizes. Among the three SDA statistics, SDA-CvM and SDA-$\chi^2$ perform similarly and consistently exhibit higher power than SDA-KS. In all settings, the SDA-based methods outperform the HP method.

\subsubsection{Multiple hypothesis testing} \label{subsec:result3}
In this section, we study the performance of the proposed multiple hypothesis testing procedure introduced in Section \ref{subsec:MH}. Based on the simulation results in Section \ref{subsec:result2}, which show that SDA-CvM and SDA-$\chi^2$ outperform SDA-KS, we focus on SDA-CvM as the test statistic. We consider two feature statistics: coefficient difference and sign-max, referred to as CvMCD-SDA and CvMSM-SDA, respectively. As a benchmark method, we also include the model-X knockoff procedure proposed by \cite{candes2018panning}, using the LASSO coefficient-difference statistic (LCD-Knockoff).

Histograms of the FDP and power, based on 200 simulated datasets, are shown in Figure S3. Here, power is defined as the proportion of active variables correctly selected. Both CvMCD-SDA and CvMSM-SDA control the FDR (the expected value of FDP) at the nominal 0.1 level, with CvMCD-SDA being more conservative. CvMSM-SDA also achieves higher power across all settings. When compared with the LCD-Knockoff method, our proposed procedures perform better in models 2 and 3 but slightly worse in model 1 (linear model). Notably, LCD-Knockoff performs the worst in model 2, where in more than 70\% of the simulations, no variable is selected.

\subsubsection{Impact of the covariates distribution} \label{subsec:result4}
Our proposed method relies on the normality assumption for the covariates $\bfX$. In this section, we evaluate its robustness under alternative distributions of $\bfX$. We consider three multivariate $t$-distributions and one multivariate chi-squared distribution, with detailed settings available in the supplemental material. 

Table \ref{table_dist} summarizes the simulation results across the four regression models. Our method successfully controls the type I error for $T_5^{\text{MVT}}$, $T_3^{\text{MVT}}$, and $T_5^{\text{GC}}$ in all settings, with a slight inflation observed under $\chi_5^{2,\text{GC}}$ in models 1 and 2. In terms of power, the method performs only slightly worse than the multivariate Gaussian case (Table S2) under $T_5^{\text{MVT}}$ and $T_3^{\text{MVT}}$, and similarly under $T_5^{\text{GC}}$. The results indicate that our method is robust under the elliptical family. For $\chi_5^{2,\text{GC}}$, performance is similar to the Gaussian case in models 1 and 2, but lower power is observed for $X_7$ in model 3 and for $X_1$ and $X_{11}$ in model 4.

\subsubsection{Impact of the precision matrix sparsity} \label{subsec:result5}
In this section, we investigate the impact of the precision matrix sparsity. We focus on the fixed precision matrix (block diagonal) specified in Section \ref{subsec:set}, considering two sparsity levels: $q = 5$ and $q = 10$. We examine three estimators of $\bfzeta_i$: (1) the LASSO estimator, denoted by $\hat{\bfzeta}_i$; (2) the LASSO estimator with a preliminary SIS step that reduces dimensionality to $\lfloor n/\log(n) \rfloor$, denoted by $\hat{\bfzeta}_i^{\text{SIS}}$; and (3) the ``oracle" estimator, which assumes that the active set $\cfI_i$ is known and applies least squares estimation conditional on $\cfI_i$, denoted by $\hat{\bfzeta}_i^{\text{OR}}$.

Figures \ref{FigMiss} and S4 summarize the simulation results for the variables across the five blocks. For $\hat{\bfzeta}_i$, the empirical type I error rate increases and power decreases as the number of active variables within the same block increases. When $q = 5$, incorporating the SIS step helps control the inflated type I error rate. When $q = 10$, the type I error becomes less inflated, but the power drops more sharply for $\hat{\bfzeta}_i^{\text{SIS}}$. The oracle estimator $\hat{\bfzeta}_i^{\text{OR}}$ maintains stable type I error and power across all settings. 

This simulation study highlights the limitations of the LASSO estimator in high-dimensional and non-sparse settings, where it fails to fully account for the influence of other active variables that are correlated with the target variable. Incorporating an SIS step can mitigate this issue when the sparsity level is moderate. However, as sparsity decreases further, the SIS step alone becomes insufficient, and then an alternative estimation strategy, such as SCAD or adaptive LASSO, may be considered.

\section{Gene expressions associated with Alzheimer's Disease} \label{sec:real}
The Alzheimer’s Disease Neuroimaging Initiative (ADNI) study was established to support the development of treatments for Alzheimer’s disease (AD) by tracking disease-related biomarkers over time. This longitudinal, multi-center study collected a comprehensive set of clinical, imaging, and genetic data from participants aged 55 to 90 across the United States and Canada. Participants included individuals with normal aging, mild cognitive impairment, dementia, and AD. Among other clinical variables, the ADNI dataset includes the Mini-Mental State Examination (MMSE) \citep{folstein1975mmse}, a widely used screening tool for cognitive function. While the optimal MMSE cutoff for identifying cognitive impairment remains a topic of debate \citep{chapman2016mmse, salis2023mmse}, a commonly used threshold for the diagnosis of dementia is a score of 24 or below, out of a maximum score of 30 \citep{tombaugh1992mmse, zhang2021mmse}.

In this study, we apply the proposed SDA method to the ADNI microarray gene expression data to identify genes associated with MMSE scores. The gene expression data were obtained from the ADNI database (\url{adni.loni.usc.edu}, downloaded on April 27, 2024). Although the dataset includes microarray data for 745 individuals, we restrict our analysis to the 292 individuals with both gene expression and MMSE measurements available at the same study visit.

Due to the ultra-high dimensionality, we perform an initial variable screening, specifically the SIS \citep{fan2008sis}, using Spearman's $\rho$ correlation coefficients. The sure screening property of the marginal correlation statistic guarantees that the procedure has a high probability of including all the relevant variables, under regularity conditions. We pre-select $p=2000$ probes, a relatively large number, to ensure that no relevant probes will be excluded from the screening. Note that, although we pre-select $2000$ probes, their conditional associations with the MMSE score will be evaluated given all the remaining 49,385 probes. 

Guided by our simulation findings in Section \ref{subsec:result3}, we implement Algorithm \ref{algo2} using the CvMSM-SDA test statistic on the 2,000 selected probes. To account for the ultra-high dimensionality when conditioning on the remaining probes, we apply the SIS procedure described in Section \ref{subsec:thm} with $\gamma = n / \log(n)$. 

The study reveals that, at FDR of $0.1$, the CvMSM-SDA selects 4 probes. We compare this result with existing literature and find that all 4 selected probes are known to be expressed more highly in AD patients than in normal patients. At a more liberal FDR of $0.2$, our method identifies an additional 7 probes. Among these probes, the targets of 6 probes have identified associations with AD in the literature, and the extra probe is a new finding. All of the identified probes, with literature references, are summarized in Table S3. 

\section{Discussion} \label{sec:con}
This article explores high-dimensional statistical inference leveraging the theory of sufficient dimension reduction. Specifically, we propose the sufficient dimension association, a model-free measure for the conditional dependence between each predictor and the response variable. We prove that, under regularity conditions, the asymptotic normality for the proposed SDA estimator can be achieved in high-dimensional settings when $\log(p) = o(n)$. Based on the central limit theorem proven in this paper, we construct SDA-$\chi^2$, SDA-KS and SDA-CvM test statistics along with a multiplier bootstrap algorithm for a single test. 

We also develop a knockoff-SDA method for multiple hypothesis testing with FDR control. One advantage for the proposed method is that the SDA statistics, as well as the corresponding knockoffs, can be obtained separately for each variable, which is easy for parallel computations and memory-efficient. Furthermore, our SDA procedure does not require estimating the distribution of $\bfX$, which is crucial for large-scale studies. For ultra-high dimensional data, such as the ADNI gene expression dataset, estimating a large but sparse covariance or precision matrix can be computationally challenging, due to the requirement of huge memory space to restore the large matrix and extensive computations associated with it. 

The validity of the proposed method relies on the normality of $\bfX$ and the sparsity of $\bfTheta$. Such conditions are commonly assumed when analyzing gene expression data. Normality of the gene expression data can be assumed either after normalization, as in the case of microarray data, or after both the mean-variance modeling and normalization, as in the case of RNAseq read-counts \citep{law2014voom}. Furthermore, our simulation results demonstrate that our proposed method is robust against model misspecification, especially within the elliptical family. Sparsity of gene regulatory networks is a common assumption in statistical and computational biology methods development \citep{noor2012inferring,wang2024wendy}. While a gene network as a whole may be highly complex, each gene is expected to interact strongly with only a few other members of the network. The covariance matrix of the gene expression levels is thus assumed to be sparse, and its inverse can be estimated using sparsity-based methods such as the graphical lasso \citep{wang2024wendy}. 

{\it Conservative type I error rates.} Our simulation studies indicate that the proposed method tends to produce conservative type I error rates. This issue is likely due to the use of SIR for constructing the sequence of transformation functions, where using indicator functions to capture local conditional dependence between $Y$ and $Z_i$ may result in information loss, particularly favoring the null hypothesis. In future work, we may explore alternative approaches, such as splines or polynomial functions, for constructing the set of transformation functions ${ g_h(\cdot) }$ to address this limitation.

{\it Survival outcome.} This work, motivated by the ADNI study, focuses on continuous outcomes. The proposed method can be easily applied to survival outcomes. Let $T$ denote the survival time, $C$ denote the censoring time, and $\Delta = I(C>T)$. Our outcome variable becomes $(Y, \Delta)$, where $Y = T(1-\Delta) + C\Delta$. By assuming $(T, C) \indep \bfX | \bfB^\top \bfX$, \cite{cook2003dimension} shows that $\cfS_{(Y, \Delta)|\bfX} \subseteq \cfS_{T|\bfX}$ and the sufficient predictors for the regression $(Y, \Delta)|\bfX$ are also sufficient predictors for the regression $T|\bfX$. By slicing the bivariate outcome $(Y, \Delta)$, we can stratify $Y$ based on the censoring indicator $\Delta$ and separately partition the range of $Y$ into $H_{\Delta=0}$ and $H_{\Delta=1}$ slices, with $H_{\Delta=0} + H_{\Delta=1} = H$. Without loss of generality, we assume balanced slices so that $n = cH$, where $|\cfJ_1| = \cdots = |\cfJ_H| = c$.

{\it Network information.} Our simulation study in Section \ref{subsec:result5} suggests that the LASSO estimator may lead to inflated type I error rates and reduced power when the precision matrix is non-sparse. This issue can be addressed by using the least squares estimator, provided that the correlation structure of $\bfX$ is known. As noted by \cite{li2008network}, network information is often available in gene expression studies, and genes connected within a network tend to exhibit similar regression coefficients. This motivates a potential future direction: incorporating network information into the proposed method. Specifically, we can modify the formulation in (\ref{lasso}) by including only the variables that are connected to $X_i$ within the network, and replacing the $l_1$-penalty with the $l_2$-penalty. Furthermore, methods for jointly testing the significance of groups of variables based on network connections can also be developed.


\section*{Acknowledgment}
The content is solely the responsibility of the authors and does not necessarily represent the official views of the National Institutes of Health. Data used in preparation of this article were obtained from the Alzheimer’s Disease Neuroimaging Initiative (ADNI) database (\url{http://adni.loni.usc.edu}). As such, the investigators within the ADNI contributed to the design and implementation of ADNI and/or provided data but did not participate in analysis or writing of this report. A complete listing of ADNI investigators can be found at: \url{http://adni.loni.usc.edu/wp-content/uploads/how_to_apply/ADNI_Acknowledgement_List.pdf}

\section*{Supplemental Material}
In the supplemental material, we provide proofs of theorems, additional simulation results and discussions referenced in Section \ref{sec:sim} and additional results referenced in Section \ref{sec:real}. 

\newpage
\begin{table}[h!] \scriptsize
\caption{Empirical power and Type I error rates for fixed precision matrix. }\label{table_compare}
\centering
\begin{tabular}{ccc|ccccccccc} 
\toprule
\hline
$n$ & $p$ & Method & $X_1$ & $X_2$ & $X_6$ & $X_7$ & $X_{11}$ & $X_{12}$ & $X_{16}$ & $X_{17}$ & Null \\ 
\hline
\multicolumn{12}{c}{Model 1} \\
\hline
400 & 1000 & SDA-KS & 0.621 & *** & 0.979 & *** & 0.999 & 1.000 & 1.000 & 1.000 & 0.030 \\
& & SDA-CvM & 0.719 & *** & 0.997 & *** & 1.000 & 1.000 & 1.000 & 1.000 & 0.044 \\
& & SDA-$\chi^2$ & 0.781 & *** & 0.998 & *** & 1.000 & 1.000 & 1.000 & 1.000 & 0.042 \\
& & HP & 0.514 & *** & 0.644 & *** & 0.675 & 0.670 & 0.714 & 0.724 & 0.073 \\
200 & 1000 & SDA-KS & 0.301 & *** & 0.750 & *** & 0.889 & 0.868 & 0.993 & 0.996 & 0.030 \\
& & SDA-CvM & 0.417 & *** & 0.826 & *** & 0.951 & 0.932 & 0.998 & 0.999 & 0.047 \\
& & SDA-$\chi^2$ & 0.429 & *** & 0.852 & *** & 0.963 & 0.950 & 1.000 & 1.000 & 0.047 \\
& & HP & 0.126 & *** & 0.300 & *** & 0.385 & 0.387 & 0.558 & 0.556 & 0.026 \\
400 & 2000 & SDA-KS & 0.634 & *** & 0.981 & *** & 1.000 & 0.998 & 1.000 & 1.000 & 0.036 \\
& & SDA-CvM & 0.751 & *** & 0.988 & *** & 1.000 & 1.000 & 1.000 & 1.000 & 0.052 \\
& & SDA-$\chi^2$ & 0.794 & *** & 0.995 & *** & 1.000 & 1.000 & 1.000 & 1.000 & 0.051 \\
& & HP & 0.256 & *** & 0.562 & *** & 0.624 & 0.612 & 0.706 & 0.707 & 0.037 \\
200 & 2000 & SDA-KS & 0.310 & *** & 0.764 & *** & 0.884 & 0.881 & 0.994 & 0.991 & 0.034 \\
& & SDA-CvM & 0.432 & *** & 0.862 & *** & 0.939 & 0.931 & 0.998 & 0.997 & 0.055 \\
& & SDA-$\chi^2$ & 0.447 & *** & 0.890 & *** & 0.952 & 0.948 & 1.000 & 0.998 & 0.056 \\
& & HP & 0.076 & *** & 0.144 & *** & 0.172 & 0.184 & 0.269 & 0.258 & 0.022 \\
\hline
\multicolumn{12}{c}{Model 2} \\
\hline
400 & 1000 & SDA-KS & 0.218 & *** & 0.598 & *** & 0.803 & 0.819 & 0.980 & 0.976 & 0.025 \\
& & SDA-CvM & 0.300 & *** & 0.739 & *** & 0.905 & 0.913 & 0.995 & 0.986 & 0.037 \\
& & SDA-$\chi^2$ & 0.318 & *** & 0.764 & *** & 0.922 & 0.925 & 0.997 & 0.993 & 0.036 \\
& & HP & 0.508 & *** & 0.666 & *** & 0.674 & 0.681 & 0.729 & 0.751 & 0.074 \\
200 & 1000 & SDA-KS & 0.103 & *** & 0.289 & *** & 0.388 & 0.389 & 0.681 & 0.689 & 0.022 \\
& & SDA-CvM & 0.163 & *** & 0.379 & *** & 0.508 & 0.515 & 0.772 & 0.762 & 0.036 \\
& & SDA-$\chi^2$ & 0.166 & *** & 0.396 & *** & 0.519 & 0.528 & 0.813 & 0.807 & 0.035 \\
& & HP & 0.107 & *** & 0.328 & *** & 0.408 & 0.387 & 0.568 & 0.575 & 0.027 \\
400 & 2000 & SDA-KS & 0.214 & *** & 0.620 & *** & 0.802 & 0.816 & 0.976 & 0.968 & 0.026 \\
& & SDA-CvM & 0.334 & *** & 0.776 & *** & 0.901 & 0.902 & 0.996 & 0.991 & 0.037 \\
& & SDA-$\chi^2$ & 0.343 & *** & 0.791 & *** & 0.913 & 0.913 & 0.998 & 0.992 & 0.037 \\
& & HP & 0.277 & *** & 0.557 & *** & 0.620 & 0.634 & 0.720 & 0.733 & 0.038 \\
200 & 2000 & SDA-KS & 0.108 & *** & 0.295 & *** & 0.408 & 0.415 & 0.665 & 0.667 & 0.025 \\
& & SDA-CvM & 0.170 & *** & 0.394 & *** & 0.518 & 0.533 & 0.775 & 0.779 & 0.038 \\
& & SDA-$\chi^2$ & 0.163 & *** & 0.413 & *** & 0.533 & 0.555 & 0.806 & 0.803 & 0.039 \\
& & HP & 0.065 & *** & 0.154 & *** & 0.182 & 0.171 & 0.254 & 0.272 & 0.022 \\
\hline
\multicolumn{12}{c}{Model 3} \\
\hline
400 & 1000 & SDA-KS & 0.926 & 0.818 & 1.000 & 0.404 & 1.000 & 0.967 & *** & *** & 0.021 \\
& & SDA-CvM & 0.954 & 0.943 & 1.000 & 0.520 & 1.000 & 0.997 & *** & *** & 0.030 \\
& & SDA-$\chi^2$ & 0.970 & 0.939 & 1.000 & 0.553 & 1.000 & 0.997 & *** & *** & 0.028 \\
& & HP & 0.717 & 0.024 & 0.960 & 0.045 & 1.000 & 0.516 & *** & *** & 0.034 \\
200 & 1000 & SDA-KS & 0.649 & 0.368 & 1.000 & 0.153 & 0.987 & 0.658  & *** & *** & 0.017 \\
& & SDA-CvM & 0.753 & 0.628 & 1.000 & 0.228 & 1.000 & 0.853 & *** & *** & 0.029 \\
& & SDA-$\chi^2$ & 0.784 & 0.611 & 1.000 & 0.227 & 0.999 & 0.845 & *** & *** & 0.028 \\
& & HP & 0.297 & 0.007 & 0.681 & 0.013 & 0.822 & 0.444 & *** & *** & 0.021 \\
400 & 2000 & SDA-KS & 0.925 & 0.818 & 1.000 & 0.396 & 1.000 & 0.984 & *** & *** & 0.021 \\
& & SDA-CvM & 0.967 & 0.944 & 1.000 & 0.547 & 1.000 & 0.998 & *** & *** & 0.031 \\
& & SDA-$\chi^2$ & 0.977 & 0.938 & 1.000 & 0.555 & 1.000 & 0.999 & *** & *** & 0.030 \\
& & HP & 0.543 & 0.010 & 0.863 & 0.043 & 0.978 & 0.511 & *** & *** & 0.028 \\
200 & 2000 & SDA-KS & 0.690 & 0.393 & 1.000 & 0.211 & 0.994 & 0.687 & *** & *** & 0.019 \\
& & SDA-CvM & 0.788 & 0.615 & 1.000 & 0.276 & 0.999 & 0.879 & *** & *** & 0.032 \\
& & SDA-CvM & 0.819 & 0.605 & 1.000 & 0.284 & 0.999 & 0.866 & *** & *** & 0.031 \\
& & HP & 0.228 & 0.004 & 0.442 & 0.024 & 0.516 & 0.248 & *** & *** & 0.015 \\
\hline
\multicolumn{12}{c}{Model 4} \\
\hline
400 & 1000 & SDA-KS & 0.414 & *** & 0.899 & *** & 0.999 & 1.000 & 0.941 & 0.989 & 0.020 \\
& & SDA-CvM & 0.459 & *** & 0.940 & *** & 1.000 & 1.000 & 0.977 & 0.994 & 0.029 \\
& & SDA-$\chi^2$ & 0.518 & *** & 0.961 & *** & 1.000 & 1.000 & 0.984 & 0.996 & 0.028 \\
& & HP & 0.696 & *** & 0.946 & *** & 1.000 & 0.519 & 0.496 & 0.500 & 0.012 \\
200 & 1000 & SDA-KS & 0.200 & *** & 0.586 & *** & 0.942 & 0.994 & 0.700 & 0.870 & 0.017 \\
& & SDA-CvM & 0.263 & *** & 0.680 & *** & 0.977 & 0.998 & 0.825 & 0.940 & 0.028 \\
& & SDA-$\chi^2$ & 0.278 & *** & 0.723 & *** & 0.990 & 0.998 & 0.826 & 0.948 & 0.027 \\
& & HP & 0.305 & *** & 0.635 & *** & 0.803 & 0.442 & 0.287 & 0.378 & 0.007 \\
400 & 2000 & SDA-KS & 0.412 & *** & 0.901 & *** & 0.997 & 1.000 & 0.950 & 0.997 & 0.020 \\
& & SDA-CvM & 0.472 & *** & 0.929 & *** & 0.999 & 1.000 & 0.979 & 1.000 & 0.029 \\
& & SDA-$\chi^2$ & 0.529 & *** & 0.962 & *** & 0.999 & 1.000 & 0.990 & 1.000 & 0.028 \\
& & HP & 0.514 & *** & 0.858 & *** & 0.980 & 0.502 & 0.448 & 0.478 & 0.008 \\
200 & 2000 & SDA-KS & 0.217 & *** & 0.618 & *** & 0.945 & 0.991 & 0.744 & 0.885 & 0.018 \\
& & SDA-CvM & 0.297 & *** & 0.691 & *** & 0.989 & 1.000 & 0.838 & 0.953 & 0.029 \\
& & SDA-$\chi^2$ & 0.312 & *** & 0.729 & *** & 0.993 & 1.000 & 0.856 & 0.958 & 0.026 \\
& & HP & 0.230 & *** & 0.445 & *** & 0.514 & 0.250 & 0.146 & 0.169 & 0.009 \\
\bottomrule
\end{tabular}
\end{table}

\begin{table}[h!] \scriptsize
\caption{Empirical power and Type I error rates with respect to different distributions of $
\bfX$. }\label{table_dist}
\centering
\begin{tabular}{c|ccccccccc} 
\toprule
\hline
Dist. & $X_1$ & $X_2$ & $X_6$ & $X_7$ & $X_{11}$ & $X_{12}$ & $X_{16}$ & $X_{17}$ & Null \\ 
\hline
\multicolumn{10}{c}{Model 1} \\
\hline
$T_5^{\text{MVT}}$ & 0.725 & *** & 0.992 & *** & 1.000 & 0.999 & 1.000 & 1.000 & 0.040 \\
$T_3^{\text{MVT}}$ & 0.681 & *** & 0.986 & *** & 0.998 & 1.000 & 1.000 & 1.000 & 0.045 \\
$T_5^{\text{GC}}$ & 0.787 & *** & 0.999 & *** & 1.000 & 1.000 & 1.000 & 1.000 & 0.044 \\
$\chi_5^{2,\text{GC}}$ & 0.873 & *** & 1.000 & *** & 1.000 & 1.000 & 1.000 & 1.000 & 0.064 \\
\hline
\multicolumn{10}{c}{Model 2} \\
\hline
$T_5^{\text{MVT}}$ & 0.353 & *** & 0.779 & *** & 0.928 & 0.933 & 0.993 & 0.994 & 0.036 \\
$T_3^{\text{MVT}}$ & 0.326 & *** & 0.750 & *** & 0.897 & 0.907 & 0.993 & 0.991 & 0.034 \\
$T_5^{\text{GC}}$ & 0.362 & *** & 0.796 & *** & 0.927 & 0.922 & 0.992 & 0.994 & 0.035 \\
$\chi_5^{2,\text{GC}}$ & 0.752 & *** & 0.996 & *** & 1.000 & 0.999 & 1.000 & 1.000 & 0.055 \\
\hline
\multicolumn{10}{c}{Model 3} \\
\hline
$T_5^{\text{MVT}}$ & 0.957 & 0.847 & 1.000 & 0.373 & 1.000 & 0.968 & *** & *** & 0.030 \\
$T_3^{\text{MVT}}$ & 0.926 & 0.751 & 1.000 & 0.287 & 1.000 & 0.923 & *** & *** & 0.030 \\
$T_5^{\text{GC}}$ & 0.979 & 0.903 & 1.000 & 0.581 & 1.000 & 0.980 & *** & *** & 0.034 \\
$\chi_5^{2,\text{GC}}$ & 0.964 & 0.791 & 1.000 & 0.139 & 1.000 & 0.816 & *** & *** & 0.037 \\
\hline
\multicolumn{10}{c}{Model 4} \\
\hline
$T_5^{\text{MVT}}$ & 0.485 & *** & 0.927 & *** & 0.999 & 1.000 & 0.939 & 0.993 & 0.027 \\
$T_3^{\text{MVT}}$ & 0.467 & *** & 0.878 & *** & 0.992 & 0.999 & 0.881 & 0.974 & 0.028 \\
$T_5^{\text{GC}}$ & 0.542 & *** & 0.954 & *** & 1.000 & 1.000 & 0.978 & 0.995 & 0.027 \\
$\chi_5^{2,\text{GC}}$ & 0.025 & *** & 1.000 & *** & 0.115 & 1.000 & 0.896 & 1.000 & 0.040 \\
\bottomrule
\end{tabular}
\end{table}

\newpage

\begin{figure}[h!]
    \centering
    \includegraphics[width=1\linewidth]{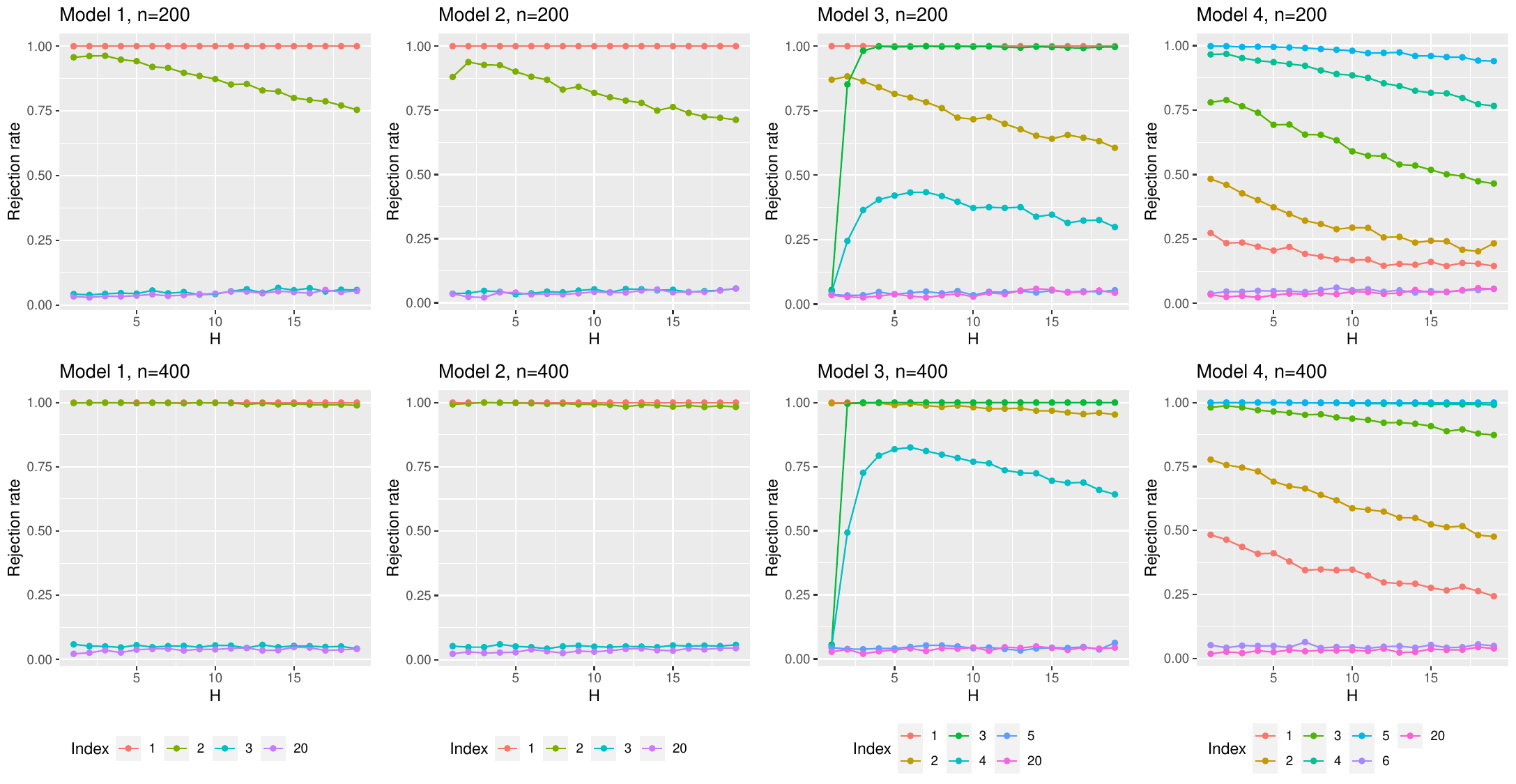}
    \caption{Empirical Type I error rates (for two selected null covariates) and power (for all active signals) with respect to $H$ for SDA-CvM.}
    \label{FigH}
\end{figure}

\begin{figure}[h!]
    \centering
    \includegraphics[width=1\linewidth]{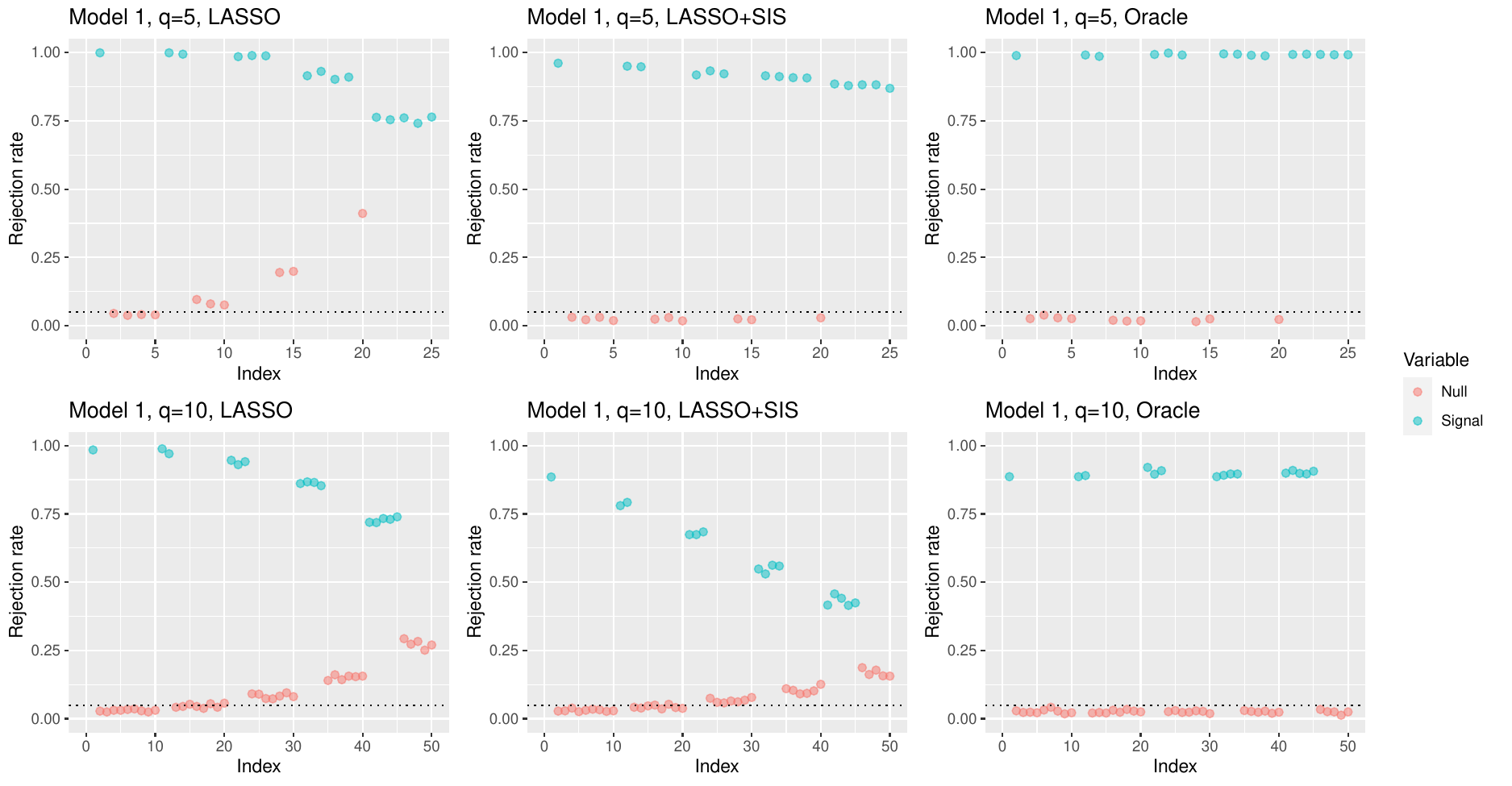}
    \caption{Empirical Type I error rates and power for different sparsity levels. }
    \label{FigMiss}
\end{figure}

\renewcommand{\theequation}{S\arabic{equation}}
\setcounter{equation}{0}

\renewcommand{\thefigure}{S\arabic{figure}}
\setcounter{figure}{0}

\renewcommand{\thetable}{S\arabic{table}}
\setcounter{table}{0}

\appendix
\begin{center}
  \appendixpagename  
\end{center}

\section{Proofs and additional discussion}
\subsection{Markov blanket and conditional dependence}
According to \cite{lauritzen1996graphical}, the Markov blanket property $Y \indep \bfX | \bfX(\cfA)$ is called the global Markov property (G), and the hypothesis $Y \indep X_i | \bfX_{-i}$ is called the pairwise Markov property (P). It is trivial to show that (G) is a sufficient condition for (P), see Proposition 3.4 in \cite{lauritzen1996graphical}. 


To show that (P) leads to (G), \cite{pearl1987logic} gave an {\it intersection} assumption: $Y \indep X_i | \bfX_{-i}$ and $Y \indep X_{i'} | \bfX_{-i'}$ leads to $Y \indep X_i \cup X_{i'} | \bfX_{-\{ i,i' \}}$ for every $i\neq i'$. This assumption does not hold universally. However, Proposition 3.1 in \cite{lauritzen1996graphical} states that if the joint density of all variables with respect to a product measure is positive and continuous, then the intersection assumption holds. In such a case, (G) is equivalent to (P). 

\subsection{Proof of Lemma 1}
We denote the left-hand side of (7) as $\bfU_i = (U_{1i}, \cdots, U_{Hi})^\top$, where
\be \label{Uh}
U_{hi} = \frac{(\hat{\bfzeta}-\bfzeta^0)^\top}{\sqrt{n}} \sum_{j=1}^n I(Y_j \in \cfJ_h) \bfx_{-i,j}.
\ee
With $l \in \cfP_i$, where $\cfP_i = \{ 1,\ldots, i-1,i+1,\ldots,p \}$, since $X_{l,j}$ and $I(Y_j \in \cfJ_h)$ are sub-Gaussian random variables, we have $I(Y_j \in \cfJ_h) X_{l,j}$ is sub-exponential \citep{wainwright2019high}. Using Proposition 2.9 in \cite{wainwright2019high}, for every $t>0$, we have the following tail bound:
\ba
P \left( \left\lvert \frac{1}{n} \sum_{j=1}^n I(Y_j \in \cfJ_h) X_{l,j} \right\rvert \ge t \right) \le 2 \exp \left\{ -C_1 n \min \left( \frac{t^2}{A_l^2}, \frac{t}{A_l} \right) \right\},
\ea
where $C_1, A_l$ are positive constants. Let $A = \max_{l \in \cfP_i} A_l$, we have
\be \label{bound_inf}
P \left( \frac{1}{n} \left\lVert \sum_{j=1}^n I(Y_j \in \cfJ_h) \bfX_{j} \right\rVert_{\infty} \ge t \right) \le 2 (p-1) \exp \left\{ -C_1 n \min \left( \frac{t^2}{A^2}, \frac{t}{A} \right) \right\}.
\ee
Thus, with $t = C_2 \sqrt{\log(p-1)/n}$, where $C_2$ is some constant, the tail bound (\ref{bound_inf}) implies that
\ba
\left\lVert \frac{1}{n}  \sum_{j=1}^n I(Y_j \in \cfJ_h) \bfX_{j} \right\rVert_{\infty} = O_p(\sqrt{\log(p)/n}).
\ea
Finally, because
\ba
\left\lVert \bfU_i \right\rVert_1 \le \left( \sum_{h=1}^H \left\lVert \frac{1}{\sqrt{n}} \sum_{j=1}^n I(Y_j \in \cfJ_h) \bfX_{j} \right\rVert_{\infty} \right) \left\lVert \hat{\bfzeta}-\bfzeta^0 \right\rVert_1,
\ea
the lemma is proved.

\subsection{Proof of Lemma 2}
Because
\ba
\frac{1}{2n} \lVert \bfx_i - \hat{\bfzeta}_i^\top \bfx_{-i} \rVert_2^2 + \lambda_i \lVert \hat{\bfzeta}_i \rVert_1 \le \frac{1}{2n} \lVert \bfx_i - \bfzeta_i^{0,\top} \bfx_{-i} \rVert_2^2 + \lambda_i \lVert \bfzeta_i^0 \rVert_1,
\ea
we have the {\it basic inequality}:
\be \label{basic}
\frac{1}{2n} \lVert (\bfzeta^0_i - \hat{\bfzeta}_i)^\top \bfx_{-i} \rVert_2^2 + \lambda_i \lVert \hat{\bfzeta}_i \rVert_1 \le \frac{1}{n} \bfZ_i^\top (\bfzeta^0_i - \hat{\bfzeta}_i)^\top \bfx_{-i} + \lambda_i \lVert \bfzeta_i^0 \rVert_1.
\ee
Let $\cfE_i$ be the set of events defined as
\ba
\cfE_i = \left\{ \lambda_i \ge \frac{1}{n} \lVert \bfZ_i^\top \bfx_{-i} \rVert_{\infty} \right\},
\ea
(\ref{basic}) implies that, under $\cfE_i$, we have
\ba
\frac{1}{2n} \lVert (\bfzeta^0_i - \hat{\bfzeta}_i)^\top \bfx_{-i} \rVert_2^2 + \lambda_i \lVert \hat{\bfzeta}_i \rVert_1 &\le& \frac{1}{n} \lVert \bfZ_i^\top \bfx_{-i} \rVert_{\infty} \lVert \bfzeta^0_i - \hat{\bfzeta}_i \rVert_1 + \lambda_i \lVert \bfzeta_i^0 \rVert_1 \\
&\le& \lambda_i \lVert \bfzeta^0_i - \hat{\bfzeta}_i \rVert_1 + \lambda_i \lVert \bfzeta_i^0 \rVert_1,
\ea
and
\be \label{e1}
&~& \frac{1}{2n} \lVert (\bfzeta^0_i - \hat{\bfzeta}_i)^\top \bfx_{-i} \rVert_2^2 + \lambda_i (\lVert \hat{\bfzeta}_i(\cfI_i) \rVert_1 + \lVert \hat{\bfzeta}_i(\cfI_i^c) \rVert_1) \nonumber \\
&\le& \lambda_i \lVert \bfzeta^0_i(\cfI_i) - \hat{\bfzeta}_i(\cfI_i) \rVert_1 + \lambda_i \lVert \hat{\bfzeta}_i(\cfI_i^c) \rVert_1 + \lambda_i \lVert \bfzeta^0_i(\cfI_i) \rVert_1.
\ee
Because
\be \label{e2}
\lVert \hat{\bfzeta}_i(\cfI_i) \rVert_1 + \lVert \hat{\bfzeta}_i(\cfI_i^c) \rVert_1 \ge \lVert \bfzeta^0_i(\cfI_i) \rVert_1 + \lVert \hat{\bfzeta}_i(\cfI_i^c) \rVert_1 - \lVert \hat{\bfzeta}_i(\cfI_i) - \bfzeta^0_i(\cfI_i) \rVert_1,
\ee
combining (\ref{e1}) and (\ref{e2}), we have
\be \label{e3}
\frac{1}{2n} \lVert (\bfzeta^0_i - \hat{\bfzeta}_i)^\top \bfx_{-i} \rVert_2^2 + \lambda_i \lVert \hat{\bfzeta}_i(\cfI_i^c) \rVert_1 \le 2 \lambda_i \lVert \bfzeta^0_i(\cfI_i) - \hat{\bfzeta}_i(\cfI_i) \rVert_1,
\ee
which implies
\be \label{e4}
\lVert \bfzeta^0_i(\cfI_i^c) - \hat{\bfzeta}_i(\cfI_i^c) \rVert_1 \le 2 \lVert \bfzeta^0_i(\cfI_i) - \hat{\bfzeta}_i(\cfI_i) \rVert_1,
\ee
and $(\bfzeta^0_i - \hat{\bfzeta}_i) \in \cfC_2(I_i)$.
Thus, based on (\ref{e3}), we have
\be \label{L1bound}
&~& \frac{1}{2n} \lVert (\bfzeta^0_i - \hat{\bfzeta}_i)^\top \bfx_{-i} \rVert_2^2 + \lambda_i \lVert \bfzeta^0_i - \hat{\bfzeta}_i \rVert_1 \nonumber \\
&=& \frac{1}{2n} \lVert (\bfzeta^0_i - \hat{\bfzeta}_i)^\top \bfx_{-i} \rVert_2^2 + \lambda_i \lVert \bfzeta^0_i(\cfI_i) - \hat{\bfzeta}_i(\cfI_i) \rVert_1 + \lambda \lVert \hat{\bfzeta}_i(\cfI_i^c) \rVert_1 \nonumber \\
&\le& 3\lambda_i \lVert \bfzeta^0_i(\cfI_i) - \hat{\bfzeta}_i(\cfI_i) \rVert_1 \nonumber \\
&\le& 3\lambda_i \sqrt{I_i} \lVert \bfzeta^0_i - \hat{\bfzeta}_i \rVert_2.
\ee
With the restricted eigenvalue condition and (\ref{e4}), we have
\be \label{L1bound2}
\lVert \bfzeta^0_i - \hat{\bfzeta}_i \rVert_2 \le \left\lVert \frac{1}{\sqrt{\kappa n}} \bfx_{-i}^\top (\bfzeta^0_i - \hat{\bfzeta}_i) \right\rVert_2.
\ee
By plugging (\ref{L1bound2}) into (\ref{L1bound}), we have
\ba
\frac{1}{2n} \left\lVert \bfx_{-i}^\top (\bfzeta^0_i - \hat{\bfzeta}_i) \right\rVert_2^2 \le 3\lambda_i \sqrt{I_i} \left\lVert \frac{1}{\sqrt{\kappa n}} \bfx_{-i}^\top (\bfzeta^0_i - \hat{\bfzeta}_i) \right\rVert_2,
\ea
which implies
\be \label{L2bound}
\frac{1}{2\sqrt{n}} \left\lVert \bfx_{-i}^\top (\bfzeta^0_i - \hat{\bfzeta}_i) \right\rVert_2 \le 3\lambda_i \sqrt{I_i} \frac{1}{\sqrt{\kappa}}.
\ee
Therefore, by plugging (\ref{L1bound2}) and (\ref{L2bound}) into (\ref{L1bound}), we have
\ba
\lVert \bfzeta^0_i - \hat{\bfzeta}_i \rVert_1 \le \frac{9}{2\kappa} \lambda_i I_i.
\ea

Finally, we need choose $\lambda_i$ so that $P(\cfE_i) \rightarrow 1$. Let $\bfX_{l} = (X_{l,1}, \cdots, X_{l,n})^\top$, because for any $t > 0$, we have
\ba
P\left( \frac{1}{n} \lVert \bfZ_i^\top \bfx_{-i} \rVert_{\infty} > t \right) &\le& \sum_{l \in \cfP_i} P\left( \frac{1}{n} \bfZ_i^\top \bfX_{l} > t \right) \\
&\le& (p-1) \exp\left( -\frac{t^2}{n^{-2} \sigma_i^2 \max_l \lVert \bfX_l \rVert_2^2 } \right) \\
&\le& (p-1) \exp\left( -\frac{n t^2}{C \sigma_i^2 } \right),
\ea
where $n^{-1} \max_l \lVert \bfX_l \rVert_2^2 \le C$. Thus, by choosing 
\ba
\lambda_i = \sqrt{\frac{C \sigma^2_i \{ \log(p - 1) + \log(\delta) \}}{n}},
\ea
where $\delta \rightarrow \infty$, we have
\ba
P(\cfE_i) = 1 - P\left( \frac{1}{n} \lVert \bfZ_i^\top \bfx_{-i} \rVert_{\infty} > \lambda_i \right) \ge 1 - \frac{1}{\delta},
\ea
and the lemma is proved.

\subsection{Proof of Lemma 3}
By the Cauchy-Schwarz inequality, we have
\ba
\E\left( \lVert \bfOmega_i^{-1/2} \bfpsi_i(\bfW) \rVert_2^3 \right) = \E\left( | \bfpsi_i^\top(\bfW) \bfOmega_i^{-1} \bfpsi_i(\bfW) |^{3/2} \right) \le \lVert \bfOmega_i^{-1} \rVert_{op}^{3/2} \E\left( \lVert \bfpsi_i(\bfW) \rVert_2^3 \right), 
\ea
where $\lVert \cdot \rVert_{op}$ is the operator norm. We can decompose the $\bfOmega_i$ as $\bfP \bfD_i - (\bfP \bfe_i) (\bfP \bfe_i)^\top$, where 
\ba
\bfP &=& \text{diag}(p_1, \cdots, p_H), \\ 
\bfD_i &=& \text{diag}(\E(Z_i^2|Y \in \cfJ_1), \cdots, \E(Z_i^2|Y \in \cfJ_H)), \\
\bfe_i &=& (\E(Z_i|Y \in \cfJ_1), \cdots, \E(Z_i|Y \in \cfJ_H))^\top.
\ea
Using the Woodbury identity, we have
\be \label{Woodbury}
\bfOmega_i^{-1} = (\bfP \bfD_i)^{-1} + \frac{(\bfP \bfD_i)^{-1} (\bfP \bfe_i) (\bfP \bfe_i)^\top (\bfP \bfD_i)^{-1}}{1 - (\bfP \bfe_i)^\top \bfD_i^{-1} (\bfP \bfe_i)}.
\ee
According to regularity conditions (C6) and (C7),
we have $\lVert \bfP \rVert_{op} = O(H^{-1})$, $\lVert \bfP^\top \rVert_{op} = O(H)$, $\lVert \bfD_i \rVert_{op} = O(H)$, $\lVert \bfD_i^\top \rVert_{op} = O(H^{-1})$, and
\ba
\lVert \bfe_i \rVert_2 = \sqrt{\sum_{h=1}^H \E(Z_i | Y \in \cfJ_h)} = O(H^{3/2}).
\ea
Thus,
\be \label{bound_op}
\lVert (\bfP \bfD_i)^{-1} \rVert_{op} &\le& \lVert \bfP^{-1} \rVert_{op} \lVert \bfD_i^{-1} \rVert_{op} = O(1), \nonumber \\
\lVert \bfP \bfe_i \rVert_2 &\le& \lVert \bfP \rVert_{op} \lVert \bfe_i \rVert_2 = O(H^{1/2}), \\
(\bfP \bfe_i)^\top \bfD_i^{-1} (\bfP \bfe_i) &\le& \lVert \bfP \bfe_i \rVert_2^2 \lVert (\bfP \bfD_i)^{-1} \rVert_{op} = O(H). \nonumber 
\ee
Thus, combining (\ref{Woodbury}) and (\ref{bound_op}), we have
\ba
\lVert \bfOmega_i^{-1} \rVert_{op} = O(1).
\ea
On the other hand, since $Z_i \sim N(0, \sigma^2)$, we have
\ba
\E\lVert \bfpsi_i(\bfW) \rVert_2^3 = \E\left( \sum_{h=1}^H [I(Y \in \cfJ_h) Z_i - \nu_{hi}^0]^2 \right)^{3/2} \le \E(Z_i^2+Z_i^4) \le C_1,
\ea
where $C_1$ is a positive constant. Therefore, we can conclude that
\be \label{bound_op3}
\E\left( \lVert \bfOmega_i^{-1/2} \bfpsi_i(\bfW) \rVert_2^3 \right) \le C_2,
\ee
where $C_2$ is a positive constant. Then, according to the Berry-Esseen Bound \citep{bentkus2005lyapunov}, we have
\be \label{Be_Es} \small
\sup_{A \in \cfC_H} \left| P\left( \frac{1}{\sqrt{n}} \sum_{j=1}^n \bfpsi_{i}(\bfw_j) \in A \right) - P\left( N(0, \bfOmega_i) \in A \right) \right| \le C \frac{H^{1/4}}{n^{1/2}} \E\left( \lVert \bfOmega_i^{-1/2} \bfpsi_i(\bfW) \rVert_2^3 \right),
\ee
where $C$ is a positive constant, and Lemma 3 can be proved by incorporating (\ref{bound_op3}) into (\ref{Be_Es}).

\subsection{Proof of Theorem 2}
First, we consider an intermediate estimator
\ba
\tilde{\bfOmega}_i = \frac{1}{n} \sum_{j=1}^n [\bfJ_j \otimes (X_{i,j} - \bfzeta_i^{0,\top} \bfX_{-i,j}) - \tilde{\bfnu}_i] [\bfJ_j \otimes (X_{i,j} - \bfzeta_i^{0,\top} \bfX_{-i,j}) - \tilde{\bfnu}_i]^\top,
\ea
where $\bfJ_j = (I(Y_j \in \cfJ_1), \cdots, I(Y_j \in \cfJ_H))^\top$, $\tilde{\bfnu}_i = (\title{\nu}_{1i}, \cdots, \tilde{\nu}_{Hi})^\top$ and
\ba
\tilde{\nu}_{hi}^\top = \frac{1}{n} \sum_{j=1}^n I(Y_j \in \cfJ_h) (X_{i,j} - \bfzeta_i^{0,\top} \bfX_{\backslash i, j}).
\ea
Under the regularity conditions (C1)-(C5), \cite{vershynin2012close} shows that
\ba
\lVert \tilde{\bfOmega}_i - \bfOmega_i^0 \rVert_{op} = o_p(1).
\ea
We then study the distance between $\tilde{\bfOmega}_i$ and $\hat{\bfOmega}_i$, where
\ba
\hat{\bfOmega}_i - \tilde{\bfOmega}_i &=& \frac{1}{n} \sum_{j=1}^n \bfJ_j \bfJ_j^\top [2 (X_{i,j} - \bfzeta_i^{0,\top} \bfX_{-i,j}) (\hat{\bfzeta}_i^{\top} \bfX_{-i,j} - \bfzeta_i^{0,\top} \bfX_{-i,j}) + (\bfzeta_i^{0,\top} \bfX_{-i,j} - \hat{\bfzeta}_i^{\top} \bfX_{-i,j})^2 ] \\ &~& + \tilde{\bfnu}_i \tilde{\bfnu}_i^\top - \hat{\bfnu}_i \hat{\bfnu}_i^\top. 
\ea
By taking the operator norm on both sides, we have
\be \label{norm}
\lVert \hat{\bfOmega}_i - \tilde{\bfOmega}_i \rVert_{op} &\le& \frac{1}{n} \sum_{j=1}^n \lVert \bfJ_j \rVert_2^2 \{2 |X_{i,j} - \bfzeta_i^{0,\top} \bfX_{-i,j}| |(\hat{\bfzeta}_i^{\top} - \bfzeta_i^{0,\top}) \bfX_{-i,j}| + |(\hat{\bfzeta}_i^{\top} - \bfzeta_i^{0,\top}) \bfX_{-i,j}|^2\} \nonumber \\ &~& + \lVert \tilde{\bfnu}_i \rVert_2^2 - \lVert \hat{\bfnu}_i \rVert_2^2.
\ee
As shown in the proof of Theorem 1, under the regularity conditions (I) and (II), we have
\ba
\lVert \tilde{\bfnu}_i \rVert_2^2 - \lVert \hat{\bfnu}_i \rVert_2^2 = o_p(1).
\ea
On the other hand, the first term in (\ref{norm}) is upper bounded by
\ba
\left( \frac{1}{n} \sum_{j=1}^n |(\hat{\bfzeta}_i^{\top} - \bfzeta_i^{0,\top}) \bfX_{-i,j}|^2 \right) \left(1 + 2 \left( \frac{1}{n} \sum_{j=1}^n |X_{i,j} - \bfzeta_i^{0,\top} \bfX_{-i,j}| \right)^{1/2} \right),
\ea
and the theorem can be proved since the first term converges to 0 as $n \rightarrow \infty$.

\subsection{Proof of Theorem 3}
Let $\bar{T}_i$ be the test statistic (either SDA-$\chi^2$, SDA-KS, or SDA-CvM) obtained by using the true values of $\bfz_i = (Z_{i,1}, \cdots, Z_{i,n})^\top$, because $M(\cdot, \cdot)$ is antisymmetric, we have $P(\bar{M}_i < -t) - P(\bar{M}_i > t) = 0$ for any $t>0$, where $\bar{M}_i := M(\bar{T}_i, \tilde{T}_i)$. Thus, as a direct consequence of Corollary 1, we have $P(M_i < -t) - P(M_i > t) \rightarrow 0$.

Next, we develop an upper-bound for the variance of $\sum_{i \in \cfA^c} I(M_i > t)$. According to Lemma 1, under regularity conditions, we have $M_i - \bar{M}_i = o_p(1)$. Thus, we have $I(M_i > t) - I(\bar{M}_i > t) = o_p(1)$, which implies
\be \label{distance1}
\Var\left( \sum_{i \in \cfA^c} I(M_i > t) \right) = \Var\left( \sum_{i \in \cfA^c} I(\bar{M}_i > t) \right) + o(p_0). 
\ee
Let
\ba
\bar{\nu}_{hi} = \frac{1}{n} \sum_{j=1}^n I(Y_j \in \cfJ_h) Z_{i,j},
\ea
because $I(Y_j \in \cfJ_h) \indep Z_{i,j}$ under $H_0$, we have $\Cov(\bar{\nu}_{hi}, \bar{\nu}_{h'i'}) = 0$ for every $h, h' \in \cfH$ and $i, i' \in \cfA^c$ with $i \ne i'$, which implies $\Cov(\bar{M}_i, \bar{M}_{i'}) = 0$. Therefore, 
\ba
\Var\left( \sum_{i \in \cfA^c} I(\bar{M}_i > t) \right) = \sum_{i \in \cfA^c} \Var(I(\bar{M}_i > t)) \le \frac{p_0}{4},
\ea
and combining with (\ref{distance1}), with a constant $C$, we have
\be \label{varbound}
\Var\left( \sum_{i \in \cfA^c} I(M_i > t) \right) \le C p_0.
\ee

Next, let
\ba
\text{FP}(t) = \sum_{i \in \cfA^c} I( M_i>t )&,& \text{FP}^0(t) = \sum_{i \in \cfA^c} P( M_i>t ), \\
\text{TP}(t) = \sum_{i \in \cfA} I( M_i>t )&,& \hat{\text{FP}}(t) = \sum_{i \in \cfA^c} P( M_i<-t ),
\ea
we show that 
\be \label{distance2}
\sup_{t \in \mathbb{R}^+} p_0^{-1} |\text{FP}(t) - \text{FP}^0(t)| \rightarrow 0, ~\sup_{t \in \mathbb{R}^+} p_0^{-1} |\hat{\text{FP}}(t) - \text{FP}^0(t)| \rightarrow 0.
\ee
For any $\delta > 0$, let $\{ t_k \}_{k=0}^{N_\delta}$ be an increasing sequence of constants satisfies $t_0 = 0$, $N_{\delta} = \lceil 2/\delta \rceil$, $t_{N_\delta} = \infty$, and $p_0^{-1} |\text{FP}^0(t_{k-1}) - \text{FP}^0(t_k)| \le \delta/2$ for every $k \ge 1$. According to the union bound, we have
\ba
P\left( \sup_{t \in \mathbb{R}^+} p_0^{-1} |\text{FP}(t) - \text{FP}^0(t)| > \delta \right) &\le& P\left( \bigcup_{k=1}^{N_\delta} \sup_{t \in [t_{k-1}, t_k)} p_0^{-1} |\text{FP}(t) - \text{FP}^0(t)| > \delta \right) \\ 
&\le& \sum_{k=1}^{N_\delta} P\left( \sup_{t \in [t_{k-1}, t_k)} p_0^{-1} |\text{FP}(t) - \text{FP}^0(t)| > \delta \right) \\
&\le& \sum_{k=1}^{N_\delta} P\left( p_0^{-1} |\text{FP}(t_{k-1}) - \text{FP}^0(t_k)| > \delta \right) \\
&\le& \sum_{k=1}^{N_\delta} P\left( p_0^{-1} |\text{FP}(t_k) - \text{FP}^0(t_k)| > \delta/2 \right).
\ea
Therefore, based on Chebyshev's inequality and the variance bound (\ref{varbound}), we have
\ba
P\left( \sup_{t \in \mathbb{R}^+} p_0^{-1} |\text{FP}(t) - \text{FP}^0(t)| > \delta \right) \le \frac{4 C N_\delta}{p_0 \delta^2} \rightarrow 0
\ea
as $p_0 \rightarrow \infty$, and the first part of (\ref{distance2}) is proved. The second part of (\ref{distance2}) can be proved in a similar manner using the asymptotic symmetric property of $M_i$ under $H_0$.

Notice that the FDP with a given threshold value $t>0$ is defined as
\ba
\text{FDP}(t) = \frac{\text{FP}(t)}{\text{FP}(t) + \text{TP}(t)} = \frac{p_0^{-1} \text{FP}(t)}{p_0^{-1} \text{FP}(t) + p_0^{-1} \text{TP}(t)}.
\ea
We further denote
\ba
\hat{\text{FDP}}(t) = \frac{p_0^{-1} \hat{\text{FP}}(t)}{p_0^{-1} \text{FP}(t) + p_0^{-1} \text{TP}(t)} ~\text{and}~ \text{FDP}^0(t)= \frac{p_0^{-1} \text{FP}^0(t)}{p_0^{-1} \text{FP}^0(t) + p_0^{-1} \text{TP}(t)},
\ea
and define $t_\delta$ such that $P(\text{FDP}(t_\delta) \le q-\delta) \rightarrow 1$ for $0< \delta < q$. From (\ref{distance2}), we have
\be \label{converg1} 
P(\hat{\text{FDP}}(t_\delta) \le q) \ge P(\text{FDP}(t_\delta) \le q-\delta) P(|\hat{\text{FDP}}(t_\delta) -\text{FDP}(t_\delta)| \le \delta) \rightarrow 1,
\ee
which implies $P(|\text{FDP}(t_\delta) - \text{FDP}(\tau)| > \delta) \rightarrow 0$. Hence,
\ba
P(\text{FDP}(\tau) \le q) \ge P(\text{FDP}(t_\delta) \le q-\delta) P(|\text{FDP}(\tau) - \text{FDP}(t_\delta)| \le \delta) \rightarrow 1,
\ea
and the first part of Theorem 3 is proved. Based on the definition of $\tau$ and (\ref{converg1}), we have
\ba
P(t_\delta \ge \tau) \ge P(\hat{\text{FDP}}(t_\delta) \le q) \rightarrow 1,
\ea 
as $p_0 \rightarrow \infty$. Therefore, according to (\ref{distance2}), we have
\ba
\limsup_{p_0 \rightarrow \infty} \E[\text{FDP}(\tau)] &\rightarrow& \limsup_{p_0 \rightarrow \infty} \E[\text{FDP}(\tau) | \tau \le t_\delta] \\
&\le& \limsup_{p_0 \rightarrow \infty} \E[\text{FDP}(\tau) - \text{FDP}^0(\tau) | \tau \le t_\delta] \\ 
&~& + \limsup_{p_0 \rightarrow \infty} \E[\text{FDP}^0(\tau) - \hat{\text{FDP}}(\tau) | \tau \le t_\delta] + \limsup_{p_0 \rightarrow \infty} \E[\hat{\text{FDP}}(\tau) | \tau \le t_\delta] \\
&\le& \limsup_{p_0 \rightarrow \infty} \E\left[ \sup_{t \in (0, t_\delta)} |\text{FDP}(\tau) - \text{FDP}^0(\tau)| \right] \\
&~& + \limsup_{p_0 \rightarrow \infty} \E\left[ \sup_{t \in (0, t_\delta)} |\text{FDP}^0(\tau) - \hat{\text{FDP}}(\tau)| \right] + \limsup_{p_0 \rightarrow \infty} \E[\hat{\text{FDP}}(\tau)] \\
&\rightarrow& \limsup_{p_0 \rightarrow \infty} \E[\hat{\text{FDP}}(\tau)] \le q,
\ea
and the second part of Theorem 3 is proved.

\section{Multiplier bootstrap}
 The multiplier bootstrap (MB) method allows us to simulate the distributions of $\hat{T}_i^{\text{KS}}$ and $\hat{T}_i^{\text{CvM}}$ under $H_0$, upon which $p$-values or critical values can be estimated. Compared to the conventional bootstrap method, the computational advantage of MB is that it does not require computing the estimator repetitively, which can be expensive for the LASSO estimator $\hat{\bfzeta}_i$.

As shown in Corollary 1, $\hat{\bfnu}_i$ is asymptotically linear with an influence function $\bfpsi_i(\bfW)$. 
Let $U_1, \cdots, U_n$ be independent and identically distributed standard normal random variables that are also independent of the data. We define a simulated process $\bfphi^u(\cdot)$ as
\ba
\bfphi_i^u(\bfw) = \frac{1}{\sqrt{n}} \sum_{j=1}^n U_j \hat{\bfpsi}_i(\bfW_j),
\ea
where $\bfw = (\bfW_1, \cdots, \bfW_n)^\top$. We then show that conditioning on $\bfw$, the simulated process $\bfphi_i^u(\bfw)$ can be used to approximate the distribution of $\sqrt{n} (\hat{\bfnu}_i - \bfnu_i^0)$.
\begin{theorem} \label{thm3}
Under regularity conditions (C3)-(C7), we have
\ba
\bfphi_i^u(\bfw) | \bfw \xrightarrow{d} N(0, \bfOmega_i^0).
\ea
\end{theorem}
{\it Proof.} Conditioning on $\bfw$, since $U_j \sim N(0, 1)$ we have
\ba
\bfphi_i^u(\bfw) | \bfw \sim N(0, \hat{\bfOmega}_i),
\ea
where $\hat{\bfOmega}_i$ is the sample variance estimator defined in (18). The theorem can then be proved based on the consistency of $\hat{\bfOmega}_i$ proved in Theorem 2.

Thus, we can simulate the null distribution by generating $L$ simulated processes $\phi_i^u(\bfw)$, and the corresponding critical value or $p$-value for the hypothesis testing can be estimated using the ``plug-in asymptotic'' method introduced in \cite{andrews2013inference}. Specifically, given a significance level $\alpha$, the simulated critical value in SDA-KS is defined as
\be \label{sim CV}
\hat{c_i} = \sup \left \{ q: P \left( \max_{h \in \cfH} \left| z_{hi}^u \right| \le q \right) \le 1 - \alpha \right\},
\ee
where
\ba
z_{hi}^u = \frac{\phi_{hi}^u(\bfw)}{\sqrt{\hat{\Omega}_i(h,h)}}.
\ea
Similarly, the $p$-value in SDA-KS is estimated as
\be \label{sim p}
p\text{-value} = \hat{P} \left( \hat{T}_i^{\text{KS}} > \max_{h \in \cfH} \left| z_{hi}^u \right| \right).
\ee
For the SDA-CvM procedure, replace $\max_{h \in \cfH} \left| z_{hi}^u \right|$ with $\int \left| z_{hi}^u \right| dQ(h)$ in (\ref{sim CV}) and (\ref{sim p}) to estimate the critical value or the $p$-value. A pseudocode for the implementation of the proposed SDA-KS testing procedure is as follows. 

\begin{algorithm}
\caption{SDA-KS hypothesis testing via multiplier bootstrap} 
\label{algo}
\begin{algorithmic}[1]
\State Divide the range of $Y$ into $H$ slices
\State Calculate $\hat{\bfzeta}_i$
\For {$h = 1, \cdots, H$}
    \State Calculate $\hat{\nu}_{hi}$
\EndFor
\State Calculate the asymptotic variance estimator $\hat{\bfOmega}_i$
\State Calculate the KS test statistic $\hat{T}_i^{\text{KS}}$
\For {$l = 1, \cdots, L$}
    \State Simulate $n \times H$ samples from standard Gaussian distribution
    \State Generate the simulated process $\bfphi_i^u(\bfw)$
\EndFor
\State Estimate the critical value $c_i$ and $p$-value based on the $L$ simulated process $\bfphi_i^u(\bfw)$
\If{$\hat{T}_i^{\text{KS}} > c_i$ or $p$-value $< \alpha$}
\State Reject $H_0$
\Else 
\State Do not reject $H_0$
\EndIf
\end{algorithmic}
\end{algorithm}

\newpage
\section{Detailed simulation settings}
\subsection{The choice of $H$} \label{ssubsec:result}
We consider sample sizes $n = 200$ or 400, dimension $p = 1000$, and covariates $X$ generated from the small-world network correlation structure, with $H$ ranging from 2 to 20. For the four regression models, the active covariate sets and corresponding coefficients are specified as follows: for models 1 and 2, we set $\cfA = \{1, 2\}$ and $\bfb(\cfA) = (1, 0.5)^\top$; for model 3, $\cfA_1 = \{1, 2\}$, $\cfA_2 = \{3, 4\}$, $\bfb_1(\cfA_1) = (1, -0.5)^\top$, and $\bfb_2(\cfA_2) = (-1, 0.5)^\top$; and for model 4, with $d = 5$, $\cfA_k = {k}$, $\bfb_1(\cfA_1) = -0.5$, and $\bfb_k(\cfA_k) = 0.25 \times k$ for $k = 2, 3, 4, 5$. Here, $\cfA_k$ denotes the index set of non-zero coefficients in $\bfb_k$, and $\bigcup_{k=1}^d \cfA_k = \cfA$. Empirical selection rates for all active signals $i \in \cfA$ (demonstration of power) and two selected null covariates (demonstration of type I error) are calculated based on 1000 simulated data sets.

\subsection{Empirical selection rates} \label{ssubsec:result2}
For the four regression models, the active variables and corresponding coefficients are specified as follows: for models 1 and 2, we set $\cfA = \{1, 6, 11, 12, 16, 17\}$ and $\bfb(\cfA) = (-0.4, 0.6, -0.8, -0.8, 1, 1)^\top$; for model 3, $\cfA_1 = \{1, 6, 11\}$, $\cfA_2 = \{2, 7, 12\}$, $\bfb_1(\cfA_1) = (0.5, -1, 0.8)^\top$, and $\bfb_2(\cfA_2) = (-0.8, -0.5, 1)^\top$; and for model 4, with $d = 5$, $\cfA_1 = \{1\}$, $\cfA_2 = \{6\}$, $\cfA_3 = \{11\}$, $\cfA_4 = \{12\}$, $\cfA_5 = \{16, 17\}$, $\bfb_1(\cfA_1) = -0.5$, $\bfb_2(\cfA_2) = 0.8$, $\bfb_3(\cfA_3) = -1$, $\bfb_4(\cfA_4) = 1.25$, and $\bfb_5(\cfA_5) = (-0.8, 1)^\top$.

\subsection{Multiple hypothesis testing} \label{ssubsec:result3}
We focus on the random network covariates, with $n=400$ and $p=1000$. We consider models 1-3 as described in Section 4.1, under two sparsity levels, $p_1 = |\cfA| = 6$ or 12. For models 1 and 2, the first $p_1$ variables are set as active, with $b_{i} = 1$ for each $i \in \cfA$. For model 3, we define $\cfA_1 = {1, \ldots, s/2}$ and $\cfA_2 = {s/2+1, \ldots, s}$, with $b_{1,i} = 1$ for $i \in \cfA_1$ and $b_{2,i'} = -1$ for $i' \in \cfA_2$.

\subsection{Impact of the covariates distribution} \label{ssubsec:result4}
We consider the following settings: (1) a multivariate $t$-distribution with degrees of freedom ($df$) 5 and precision matrix $\bfTheta$ ($T_5^{\text{MVT}}$); (2) a multivariate $t$-distribution with $df = 3$ and precision matrix $\bfTheta$ ($T_3^{\text{MVT}}$); (3) marginal distributions following a $t$-distribution with $df = 5$, and correlation structure generated from a Gaussian copula with precision matrix $\bfTheta$ ($T_5^{\text{GC}}$); and (4) marginal distributions following a chi-squared distribution with $df = 5$, and correlation structure generated from a Gaussian copula with precision matrix $\bfTheta$ ($\chi_5^{2,\text{GC}}$). As in Section \ref{ssubsec:result3}, we focus on SDA-CvM with random network covariates, $n = 400$, and $p = 1000$. The sets $\cfA$ and corresponding coefficients follow the specifications from Section \ref{ssubsec:result2}.

\subsection{Impact of the precision matrix sparsity} \label{ssubsec:result5}
We evaluate models 1 and 2 from Section 4.1. For each block $k = 1, \cdots, 5$, the first $k$ variables are set as active, with $b_i = 0.8$ for every $i \in \cfA$.

\clearpage
\newpage

\section{Additional simulation results}
\begin{table}[h!]
\caption{Selection rates for the LASSO estimator used for selective inference. }\label{table_SI}
\centering
\begin{tabular}{c|c|c|c|c|c|c|c|c}
\toprule
    \hline
    $|\cfA|$ & $n$ & $p$ & $\beta=0.2$ & $\beta=-0.4$ & $\beta=0.6$ & $\beta=0.8$ & $\beta=1.0$ & $\beta=0$ \\
    \hline
    \multicolumn{9}{c}{Model 1, $q=5$} \\
    \hline
    $5$ & 400 & 1000 & 0.987 & 0.992 & 0.996 & 0.998 & 0.997 & 0.011--0.050 \\
    & 200 & 1000 & 0.799 & 0.973 & 0.970 & 0.974 & 0.976 & 0.013--0.052 \\
    & 400 & 2000 & 0.983 & 0.992 & 0.992 & 0.989 & 0.991 & 0.006--0.039 \\
    & 200 & 2000 & 0.720 & 0.974 & 0.967 & 0.978 & 0.971 & 0.005--0.039 \\
    \hline
    25 & 400 & 1000 & 0.148 & 0.337 & 0.553 & 0.591 & 0.578 & 0.063--0.115 \\
    & 200 & 1000 & 0.019 & 0.040 & 0.066 & 0.114 & 0.191 & 0.005--0.028 \\
    & 400 & 2000 & 0.067 & 0.151 & 0.310 & 0.521 & 0.634 & 0.019--0.060 \\
    & 200 & 2000 & 0.009 & 0.016 & 0.030 & 0.059 & 0.109 & 0.000--0.018 \\
    \hline
    \multicolumn{9}{c}{Model 1, $q=10$} \\
    \hline
    5 & 400 & 1000 & 0.994 & 0.994 & 0.996 & 0.997 & 0.996 & 0.016--0.049 \\
    & 200 & 1000 & 0.837 & 0.984 & 0.976 & 0.980 & 0.979 & 0.017--0.050 \\
    & 400 & 2000 & 0.988 & 0.987 & 0.993 & 0.990 & 0.992 & 0.004--0.037 \\
    & 200 & 2000 & 0.765 & 0.955 & 0.955 & 0.960 & 0.962 & 0.005--0.034 \\
    \hline
    50 & 400 & 1000 & 0.012 & 0.017 & 0.026 & 0.047 & 0.066 & 0.002--0.025 \\
    & 200 & 1000 & 0.010 & 0.012 & 0.013 & 0.021 & 0.029 & 0.001--0.018 \\
    & 400 & 2000 & 0.006 & 0.010 & 0.017 & 0.031 & 0.045 & 0.000--0.014 \\
    & 200 & 2000 & 0.005 & 0.007 & 0.010 & 0.015 & 0.021 & 0.000--0.014 \\
    \hline
    \multicolumn{9}{c}{Model 2, $q=5$} \\
    \hline
    5 & 400 & 1000 & 0.020 & 0.075 & 0.183 & 0.292 & 0.409 & 0.000--0.016 \\
    & 200 & 1000 & 0.007 & 0.017 & 0.051 & 0.106 & 0.162 & 0.000--0.010 \\
    & 400 & 2000 & 0.008 & 0.059 & 0.146 & 0.255 & 0.361 & 0.000--0.011 \\
    & 200 & 2000 & 0.008 & 0.008 & 0.040 & 0.086 & 0.141 & 0.000--0.008 \\
    \hline
    25 & 400 & 1000 & 0.003 & 0.004 & 0.004 & 0.006 & 0.007 & 0.000--0.008 \\
    & 200 & 1000 & 0.003 & 0.001 & <0.001 & 0.003 & 0.008 & 0.000--0.007 \\
    & 400 & 2000 & 0.002 & 0.003 & 0.002 & 0.004 & 0.005 & 0.000--0.007 \\
    & 200 & 2000 & 0.002 & 0.002 & <0.001 & 0.003 & 0.006 & 0.000--0.007 \\
    \hline
    \multicolumn{9}{c}{Model 2, $q=10$} \\
    \hline
    5 & 400 & 1000 & 0.035 & 0.129 & 0.241 & 0.400 & 0.510 & 0.001--0.020 \\
    & 200 & 1000 & 0.014 & 0.029 & 0.072 & 0.143 & 0.235 & 0.000--0.013 \\
    & 400 & 2000 & 0.018 & 0.088 & 0.199 & 0.350 & 0.485 & 0.000--0.013 \\
    & 200 & 2000 & 0.010 & 0.023 & 0.048 & 0.115 & 0.188 & 0.000--0.007 \\
    \hline
    50 & 400 & 1000 & 0.003 & 0.002 & 0.003 & 0.004 & 0.004 & 0.000--0.009 \\
    & 200 & 1000 & 0.002 & 0.002 & 0.003 & 0.002 & 0.003 & 0.000--0.008 \\
    & 400 & 2000 & 0.002 & 0.002 & 0.002 & 0.003 & 0.002 & 0.000--0.007 \\
    & 200 & 2000 & 0.001 & 0.001 & 0.002 & 0.001 & 0.001 & 0.000--0.006 \\
\bottomrule
\end{tabular}
\end{table}

\begin{table}[h!] \tiny
\caption{Empirical power and Type I error rates for network covariates matrix. }\label{table_compare2}
\centering
\begin{tabular}{ccc|ccccccccc} 
\toprule
\hline
$n$ & $p$ & Method & $X_1$ & $X_2$ & $X_6$ & $X_7$ & $X_{11}$ & $X_{12}$ & $X_{16}$ & $X_{17}$ & Null \\ 
\hline
\multicolumn{12}{c}{Model 1} \\
\hline
400 & 1000 & SDA-KS & 0.527 & *** & 0.934 & *** & 0.998 & 0.998 & 1.000 & 1.000 & 0.023 \\
& & SDA-CvM & 0.645 & *** & 0.961 & *** & 1.000 & 1.000 & 1.000 & 1.000 & 0.035 \\
& & SDA-Chi & 0.675 & *** & 0.978 & *** & 1.000 & 1.000 & 1.000 & 1.000 & 0.032 \\
& & HP & 0.405 & *** & 0.620 & *** & 0.709 & 0.704 & 0.813 & 0.808 & 0.034 \\
200 & 1000 & SDA-KS & 0.257 & *** & 0.589 & *** & 0.918 & 0.919 & 0.996 & 1.000 & 0.018 \\
& & SDA-CvM & 0.350 & *** & 0.709 & *** & 0.949 & 0.949 & 0.999 & 1.000 & 0.032 \\
& & SDA-Chi & 0.367 & *** & 0.742 & *** & 0.964 & 0.964 & 0.999 & 1.000 & 0.030 \\
& & HP & 0.068 & *** & 0.209 & *** & 0.389 & 0.408 & 0.592 & 0.588 & 0.018 \\
400 & 2000 & SDA-KS & 0.523 & *** & 0.933 & *** & 0.998 & 0.998 & 1.000 & 1.000 & 0.024 \\
& & SDA-CvM & 0.632 & *** & 0.959 & *** & 0.999 & 0.999 & 1.000 & 1.000 & 0.034 \\
& & SDA-Chi & 0.657 & *** & 0.969 & *** & 1.000 & 1.000 & 1.000 & 1.000 & 0.033 \\
& & HP & 0.149 & *** & 0.452 & *** & 0.670 & 0.663 & 0.768 & 0.802 & 0.018 \\
200 & 2000 & SDA-KS & 0.261 & *** & 0.610 & *** & 0.913 & 0.896 & 0.998 & 0.999 & 0.019 \\
& & SDA-CvM & 0.356 & *** & 0.712 & *** & 0.957 & 0.937 & 0.999 & 1.000 & 0.031 \\
& & SDA-Chi & 0.377 & *** & 0.742 & *** & 0.969 & 0.952 & 0.999 & 1.000 & 0.031 \\
& & HP & 0.057 & *** & 0.113 & *** & 0.199 & 0.206 & 0.310 & 0.321 & 0.016 \\
\hline
\multicolumn{12}{c}{Model 2} \\
\hline
400 & 1000 & SDA-KS & 0.244 & *** & 0.604 & *** & 0.932 & 0.921 & 0.997 & 0.995 & 0.022 \\
& & SDA-CvM & 0.292 & *** & 0.692 & *** & 0.938 & 0.944 & 0.998 & 1.000 & 0.031 \\
& & SDA-Chi & 0.314 & *** & 0.718 & *** & 0.959 & 0.963 & 1.000 & 1.000 & 0.030 \\
& & HP & 0.418 & *** & 0.604 & *** & 0.714 & 0.703 & 0.796 & 0.807 & 0.033 \\
200 & 1000 & SDA-KS & 0.118 & *** & 0.290 & *** & 0.615 & 0.614 & 0.892 & 0.914 & 0.018 \\
& & SDA-CvM & 0.159 & *** & 0.334 & *** & 0.663 & 0.629 & 0.910 & 0.925 & 0.029 \\
& & SDA-Chi & 0.164 & *** & 0.362 & *** & 0.704 & 0.693 & 0.931 & 0.945 & 0.029 \\
& & HP & 0.079 & *** & 0.194 & *** & 0.420 & 0.417 & 0.579 & 0.583 & 0.018 \\
400 & 2000 & SDA-KS & 0.221 & *** & 0.627 & *** & 0.934 & 0.931 & 0.996 & 1.000 & 0.022 \\
& & SDA-CvM & 0.280 & *** & 0.662 & *** & 0.941 & 0.930 & 0.998 & 1.000 & 0.037 \\
& & SDA-Chi & 0.303 & *** & 0.717 & *** & 0.961 & 0.960 & 0.999 & 1.000 & 0.030 \\
& & HP & 0.171 & *** & 0.465 & *** & 0.641 & 0.654 & 0.793 & 0.779 & 0.019 \\
200 & 2000 & SDA-KS & 0.113 & *** & 0.306 & *** & 0.621 & 0.593 & 0.904 & 0.898 & 0.017 \\
& & SDA-CvM & 0.134 & *** & 0.343 & *** & 0.649 & 0.605 & 0.901 & 0.906 & 0.029 \\
& & SDA-Chi & 0.154 & *** & 0.375 & *** & 0.691 & 0.662 & 0.928 & 0.942 & 0.028 \\
& & HP & 0.050 & *** & 0.115 & *** & 0.204 & 0.184 & 0.292 & 0.313 & 0.015 \\
\hline
\multicolumn{12}{c}{Model 3} \\
\hline
400 & 1000 & SDA-KS & 0.940 & 0.929 & 1.000 & 0.422 & 1.000 & 0.992 & *** & *** & 0.021 \\
& & SDA-CvM & 0.965 & 0.992 & 1.000 & 0.627 & 1.000 & 0.999 & *** & *** & 0.031 \\
& & SDA-Chi & 0.974 & 0.984 & 1.000 & 0.618 & 1.000 & 1.000 & *** & *** & 0.030 \\
& & HP & 0.699 & 0.007 & 0.966 & 0.009 & 0.997 & 0.503 & *** & *** & 0.027 \\
200 & 1000 & SDA-KS & 0.628 & 0.459 & 1.000 & 0.164 & 0.992 & 0.695 & *** & *** & 0.017 \\
& & SDA-CvM & 0.732 & 0.728 & 1.000 & 0.284 & 0.995 & 0.908 & *** & *** & 0.029 \\
& & SDA-Chi & 0.754 & 0.710 & 1.000 & 0.263 & 0.999 & 0.893 & *** & *** & 0.030 \\
& & HP & 0.290 & 0.007 & 0.692 & 0.006 & 0.742 & 0.416 & *** & *** & 0.019 \\
400 & 2000 & SDA-KS & 0.929 & 0.936 & 1.000 & 0.406 & 1.000 & 0.994 & *** & *** & 0.021 \\
& & SDA-CvM & 0.968 & 0.992 & 1.000 & 0.639 & 1.000 & 0.999 & *** & *** & 0.031 \\
& & SDA-Chi & 0.973 & 0.988 & 1.000 & 0.623 & 1.000 & 1.000 & *** & *** & 0.029 \\
& & HP & 0.585 & 0.004 & 0.895 & 0.004 & 0.972 & 0.498 & *** & *** & 0.024 \\
200 & 2000 & SDA-KS & 0.605 & 0.458 & 1.000 & 0.142 & 0.994 & 0.702 & *** & *** & 0.018 \\
& & SDA-CvM & 0.726 & 0.737 & 1.000 & 0.257 & 0.994 & 0.913 & *** & *** & 0.030 \\
& & SDA-Chi & 0.730 & 0.730 & 1.000 & 0.301 & 0.999 & 0.881 & *** & *** & 0.028 \\
& & HP & 0.248 & 0.012 & 0.455 & 0.012 & 0.488 & 0.277 & *** & *** & 0.015 \\
\hline
\multicolumn{12}{c}{Model 4} \\
\hline
400 & 1000 & SDA-KS & 0.384 & *** & 0.884 & *** & 0.984 & 0.999 & 0.799 & 0.989 & 0.021 \\
& & SDA-CvM & 0.473 & *** & 0.932 & *** & 0.989 & 1.000 & 0.888 & 0.994 & 0.031 \\
& & SDA-CvM & 0.503 & *** & 0.945 & *** & 0.993 & 1.000 & 0.920 & 0.997 & 0.030 \\
& & HP & 0.713 & *** & 0.967 & *** & 0.997 & 0.500 & 0.465 & 0.499 & 0.006 \\
200 & 1000 & SDA-KS & 0.198 & *** & 0.547 & *** & 0.780 & 0.955 & 0.461 & 0.802 & 0.018 \\
& & SDA-CvM & 0.284 & *** & 0.651 & *** & 0.868 & 0.979 & 0.593 & 0.893 & 0.030 \\
& & SDA-Chi & 0.252 & *** & 0.671 & *** & 0.855 & 0.984 & 0.595 & 0.871 & 0.028 \\
& & HP & 0.333 & *** & 0.688 & *** & 0.753 & 0.419 & 0.222 & 0.355 & 0.007 \\
400 & 2000 & SDA-KS & 0.411 & *** & 0.864 & *** & 0.981 & 1.000 & 0.808 & 0.976 & 0.021 \\
& & SDA-CvM & 0.492 & *** & 0.923 & *** & 0.994 & 1.000 & 0.886 & 0.993 & 0.031 \\
& & SDA-Chi & 0.515 & *** & 0.945 & *** & 0.996 & 1.000 & 0.911 & 0.994 & 0.030 \\
& & HP & 0.619 & *** & 0.909 & *** & 0.963 & 0.497 & 0.410 & 0.481 & 0.004 \\
200 & 2000 & SDA-KS & 0.181 & *** & 0.532 & *** & 0.763 & 0.953 & 0.473 & 0.771 & 0.018 \\
& & SDA-CvM & 0.262 & *** & 0.635 & *** & 0.856 & 0.981 & 0.589 & 0.859 & 0.030 \\
& & SDA-Chi & 0.244 & *** & 0.650 & *** & 0.877 & 0.986 & 0.621 & 0.901 & 0.028 \\
& & HP & 0.228 & *** & 0.481 & *** & 0.517 & 0.284 & 0.184 & 0.230 & 0.007 \\
\bottomrule
\end{tabular}
\end{table}

\begin{figure}[h!]
    \centering
    \includegraphics[width=1\linewidth]{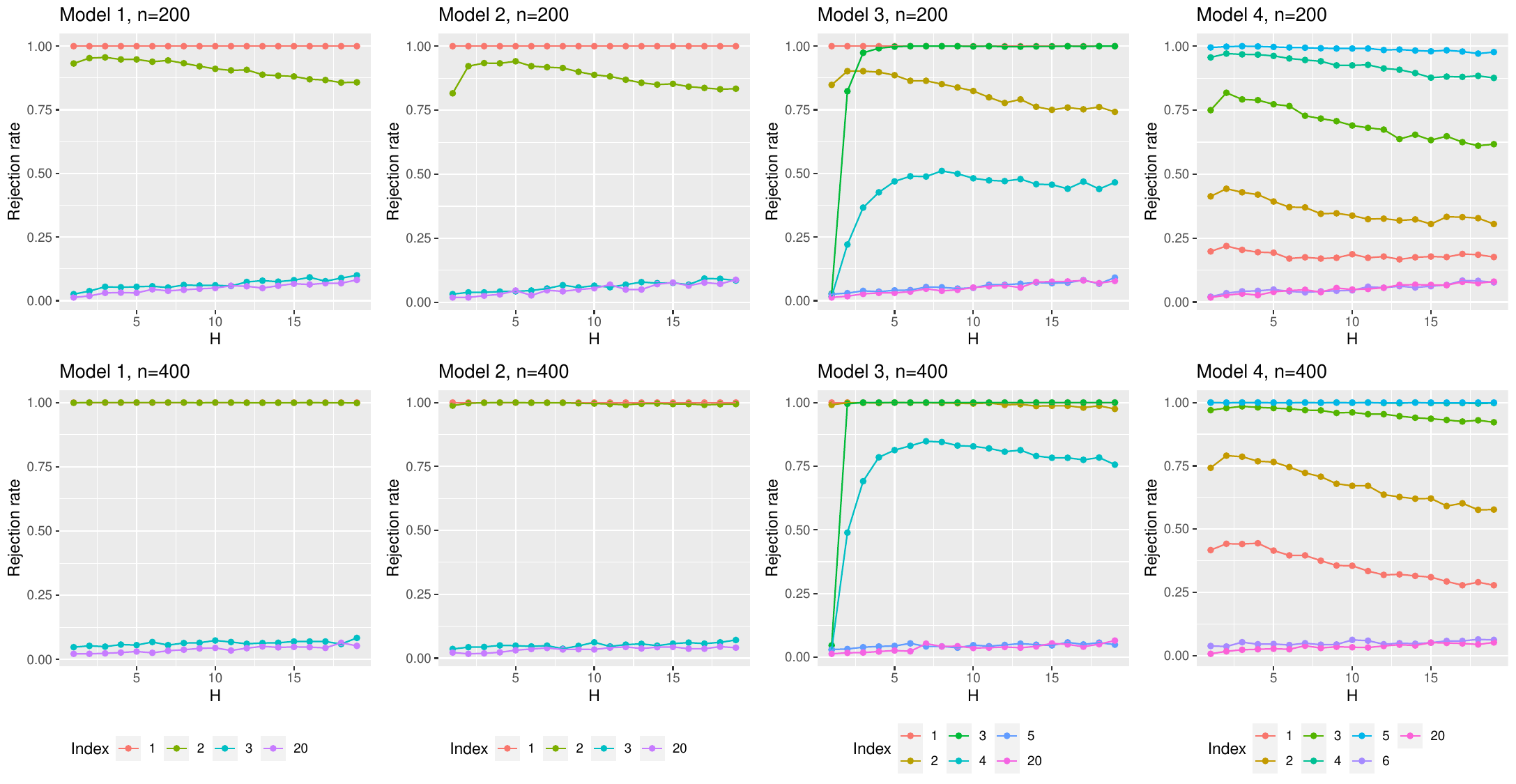}
    \caption{Empirical type I error rates and power with respect to $H$ for SDA-Chi. }
    \label{FigH2}
\end{figure}

\begin{figure}[h!]
    \centering
    \includegraphics[width=1\linewidth]{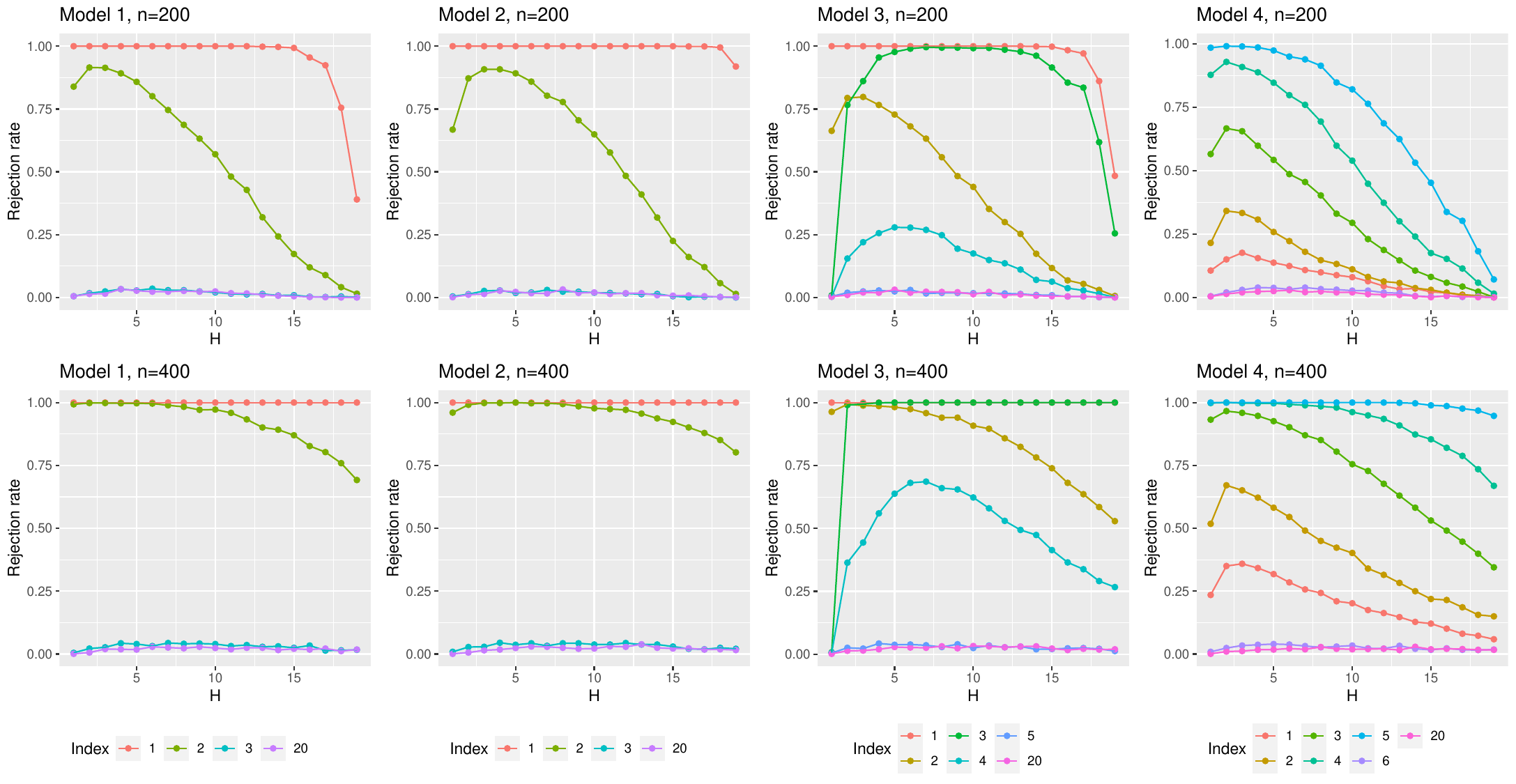}
    \caption{Empirical type I error rates and power with respect to $H$ for SDA-KS. }
    \label{FigH3}
\end{figure}

\begin{figure}[h!]
    \centering
    \includegraphics[width=1\linewidth]{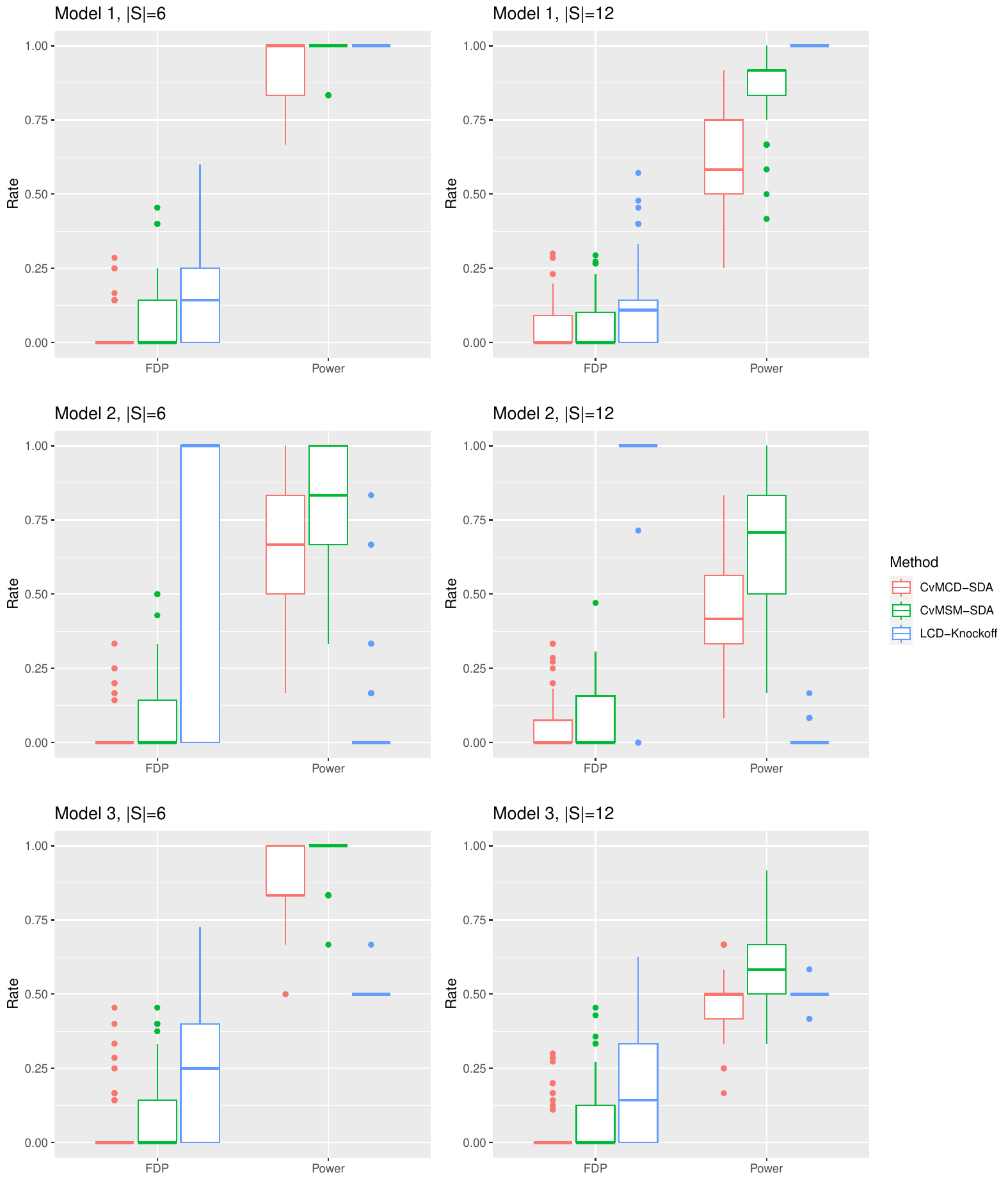}
    \caption{False discovery proportions and power. }
    \label{FigMT}
\end{figure}

\begin{figure}[h!]
    \centering
    \includegraphics[width=1\linewidth]{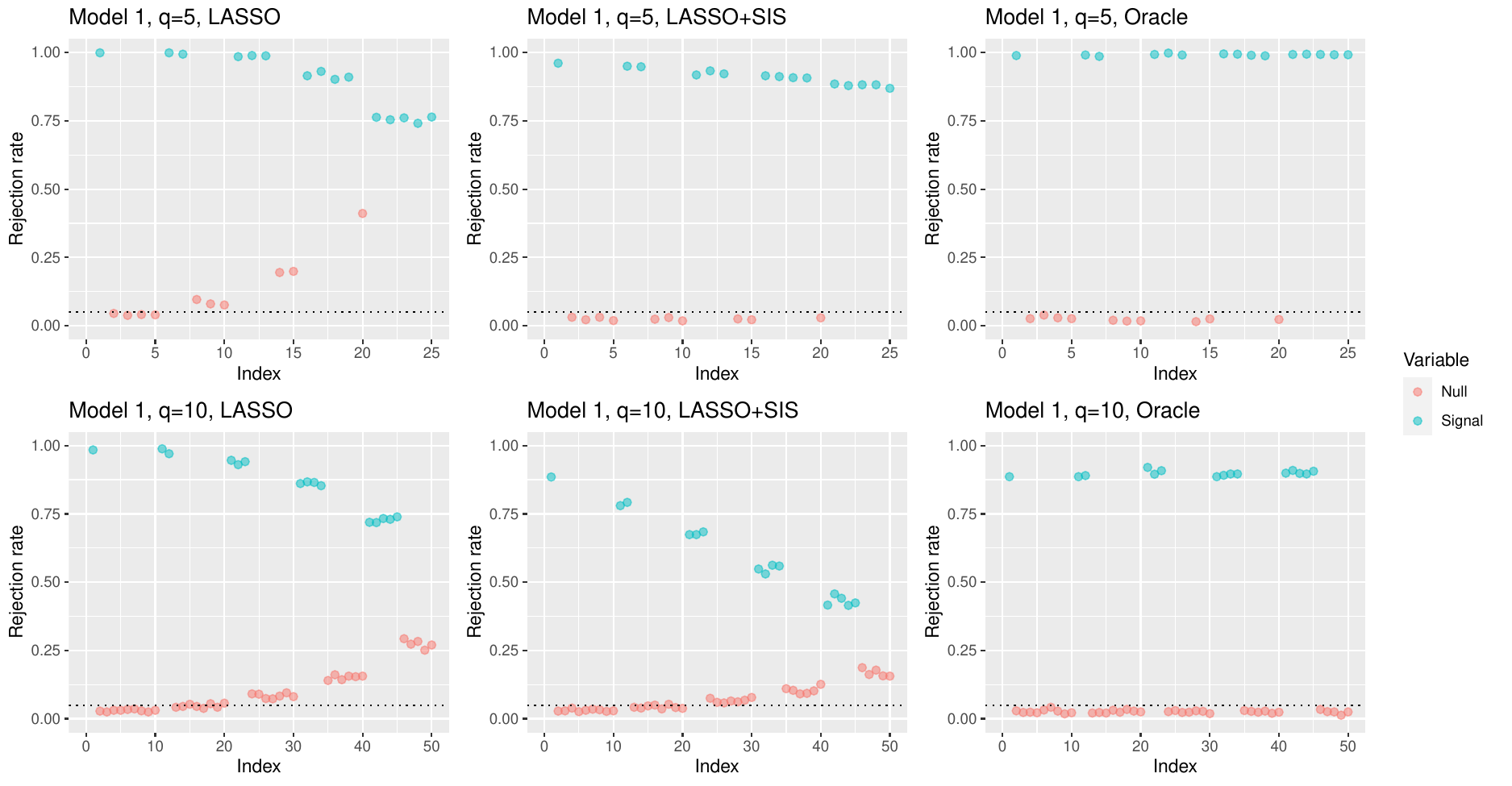}
    \caption{Empirical type I error rates and power for different sparsity levels. }
    \label{FigMiss2}
\end{figure}

\clearpage
\newpage

\section{Results for real data analysis}
\begin{table}[h!] \scriptsize
\caption{Results of CvMSM-SDA applied to ADNI gene expression data.}\label{table_results}
\centering
\begin{tabular}{l|l|l}
\toprule
    \hline
    $\textrm{Probe name}$ & $\textrm{Target description}$ & $\textrm{Confirmed by}$  \\
    \hline
    \multicolumn{3}{c}{FDR = 0.1} \\
    \hline
    11749948\_x\_at & Hydroxysteroid (17-beta) dehydrogenase 1 &  \citep{vinklarova2020} \\
 11727968\_at & Establishment of sister chromatid cohesion N-acetyltransferase 2 & \citep{wu2015}\\
 11719296\_a\_at & MAPK Associated Protein 1 &\citep{davoody2024}\\
 11715479\_a\_at & Gamma-Aminobutyric Acid Receptor-associated Protein&\citep{chen2024}\\
    \hline
    \multicolumn{3}{c}{FDR = 0.2: additional selections} \\
    \hline
    11715876\_a t& Tax-1 binding protein 3 &  -- \\
 11734725\_a\_at & Polynucleotide phosphorylase (PNPase) & \citep{hu2023}\\
 11737721\_x\_at & Collagen type XXV alpha 1 & \citep{tong2010}\\
 11721887\_a\_at & Crystallin Mu &\citep{sakkaki2024}\\
 11727893\_at & Proline And Arginine Rich End Leucine Rich Repeat Protein &\citep{mo2025}\\
 11731913\_at & G Protein-Coupled Receptor 12 & \citep{ozarslan2024}\\
 11755924\_a\_at & RAB11 family interacting protein 4 (class II)&\citep{sultana2022}\\
 \hline
 \bottomrule
\end{tabular}
\end{table}

\clearpage
\newpage

\bibliographystyle{apalike}
\bibliography{ref}

\end{document}